\documentclass{article}
\usepackage{graphicx,epsfig,color}
\newcounter{savenumi}

\newtheorem{theoremfoo}{Theorem}
\newenvironment{theorem}{\pagebreak[1]\begin{theoremfoo}}{\end{theoremfoo}}

\newtheorem{propositionfoo}[theoremfoo]{Proposition}

\newtheorem{lemmafoo}[theoremfoo]{Lemma}
\newenvironment{lemma}{\pagebreak[1]\begin{lemmafoo}}{\end{lemmafoo}}
\newtheorem{conjecturefoo}[theoremfoo]{Conjecture}

\newtheorem{corollaryfoo}[theoremfoo]{Corollary}
\newenvironment{corollary}{\pagebreak[1]\begin{corollaryfoo}}{\end{corollaryfoo}}
\newtheorem{exercisefoo}{Exercise}

\newtheorem{openfoo}[theoremfoo]{Question}

\newtheorem{nttn}[theoremfoo]{Notation}

\newtheorem{dfntn}[theoremfoo]{Definition}
\newenvironment{definition}{\pagebreak[1]\begin{dfntn}\rm}{\end{dfntn}}

\newenvironment{proof}
    {\pagebreak[1]{\narrower\noindent {\bf Proof:\quad\nopagebreak}}}{\QED}

\newcommand{\floor}[1]{\left\lfloor#1\right\rfloor}
\newcommand{\ceiling}[1]{\left\lceil#1\right\rceil}
\def\nre.{$n$\/-r.e.}
\newcommand{\scrod}{\quad\nopagebreak}

\newtheorem{factfoo}[theoremfoo]{Fact}

\newcommand{\squeeze}{
\textwidth 6in
\textheight 8.8in
\oddsidemargin 0.2in
\topmargin -0.4in
}

\newtheorem{propertyfoo}[theoremfoo]{Property}

\makeatletter

\def\@makechapterhead#1{ \vspace*{50pt} { \parindent 0pt \raggedright 
 \ifnum \c@secnumdepth >\m@ne \huge\bf \@chapapp{} \thechapter. \par 
 \vskip 20pt \fi \Huge \bf #1\par 
 \nobreak \vskip 40pt } }

\def\@sect#1#2#3#4#5#6[#7]#8{\ifnum #2>\c@secnumdepth
     \def\@svsec{}\else 
     \refstepcounter{#1}\edef\@svsec{\csname the#1\endcsname.\hskip 1em }\fi
     \@tempskipa #5\relax
      \ifdim \@tempskipa>\z@ 
        \begingroup #6\relax
          \@hangfrom{\hskip #3\relax\@svsec}{\interlinepenalty \@M #8\par}
        \endgroup
       \csname #1mark\endcsname{#7}\addcontentsline
         {toc}{#1}{\ifnum #2>\c@secnumdepth \else
                      \protect\numberline{\csname the#1\endcsname}\fi
                    #7}\else
        \def\@svsechd{#6\hskip #3\@svsec #8\csname #1mark\endcsname
                      {#7}\addcontentsline
                           {toc}{#1}{\ifnum #2>\c@secnumdepth \else
                             \protect\numberline{\csname the#1\endcsname}\fi
                       #7}}\fi
     \@xsect{#5}}

\def\@begintheorem#1#2{\it \trivlist \item[\hskip \labelsep{\bf #1\ #2.}]}

\def\@opargbegintheorem#1#2#3{\it \trivlist
      \item[\hskip \labelsep{\bf #1\ #2\ (#3).}]}

\makeatother

\newif\ifshortconferences
\shortconferencesfalse
\newif\ifmediumconferences
\mediumconferencesfalse

\def\ending#1{{\count1=#1\relax
\count2=\count1
\divide\count2 by 100
\multiply\count2 by 100
\advance\count1 by -\count2
\ifnum\count1=11
th%
\else \ifnum\count1=12
th%
\else \ifnum\count1=13
th%
\else 
\count2=\count1
\divide\count1 by 10
\multiply\count1 by 10
\advance\count2 by -\count1
\ifnum\count2=1
st%
\else \ifnum\count2=2
nd%
\else \ifnum\count2=3
rd%
\else th%
\fi\fi\fi\fi\fi\fi
}}

\def\Proceedingsofthe{\ifshortconferences Proc.\else\ifmediumconferences Proc.\else Proceedings of the\fi\fi}

\newcounter{confnum}

\def\conf#1#2{%
\setcounter{confnum}{#2}%
\addtocounter{confnum}{-\csname #1zero\endcsname}%
\ifnum\value{confnum}=1%
\expandafter\ifx\csname #1One\endcsname\relax%
\Proceedingsofthe\ \arabic{confnum}\ending{\value{confnum}}\ \csname #1name\endcsname%
\else \csname #1One\endcsname\fi%
\else%
\Proceedingsofthe\
\arabic{confnum}\ending{\value{confnum}}\ \csname #1name\endcsname\fi}

\def\qsym{\vrule width0.7ex height0.9em depth0ex}

\newif\ifqed\qedtrue

\def\noqed{\global\qedfalse}

\def\qed{\ifqed{\penalty1000\unskip\nobreak\hfil\penalty50
\hskip2em\hbox{}\nobreak\hfil\qsym
\parfillskip=0pt \finalhyphendemerits=0\par\medskip}\fi\global\qedtrue}

\makeatletter
\def\eqnqed{\noqed
	\def\@tempa{equation}
	\ifx\@tempa\@currenvir\def\@eqnnum{\qsym}%
	\addtocounter{equation}{-1}\else%
    \def\@@eqncr{\let\@tempa\relax
    \ifcase\@eqcnt \def\@tempa{& & &}\or \def\@tempa{& &}%
      \else \def\@tempa{&}\fi
     \@tempa {\def\@eqnnum{{\qsym}}\@eqnnum}
     \global\@eqnswtrue\global\@eqcnt\z@\cr}\fi}

\def\eqnlabel#1#2{\if@filesw {\let\thepage\relax%
   \def\protect{\noexpand\noexpand\noexpand}%
   \edef\@tempa{\write\@auxout{\string
      \newlabel{#2}{{{#1}}{\thepage}}}}%
   \expandafter}\@tempa%
   \if@nobreak \ifvmode\nobreak\fi\fi\fi%
	\def\@tempa{equation}
	\ifx\@tempa\@currenvir\def\theequation{{#1}}%
	\addtocounter{equation}{-1}\else%
    \def\@@eqncr{\let\@tempa\relax
    \ifcase\@eqcnt \def\@tempa{& & &}\or \def\@tempa{& &}%
      \else \def\@tempa{&}\fi
     \@tempa {\def\@eqnnum{{#1}}\@eqnnum}
     \global\@eqnswtrue\global\@eqcnt\z@\cr}\fi}

\makeatother

\def\QED{\qed}

\makeatother

\squeeze

\usepackage{enumerate}

\newcommand{\Rank}{{\rm Rank}}

\newcommand{\minRank}{{\rm minRank}}

\newcommand{\maxRank}{{\rm maxRank}}

\newcommand{\prob}{{\rm Pr}}

\newcommand{\Sopt}{{\rm Sopt}}

\newcommand{\mod}{{\rm mod}}

\begin{document}

\date{}

\title{A Dense  Hierarchy of Sublinear Time
Approximation Schemes for Bin Packing}

\author{
Richard Beigel$^1$ and Bin Fu$^2$\\\\
$^1$CIS Department\\
 Temple University,
  Philadelphia PA 19122-6094,
USA\\
beigel@cis.temple.edu\\
\\
$^2$Department of Computer Science\\
 University of Texas-Pan American,
 Edinburg, TX  78539, USA\\
binfu@cs.panam.edu\\\\
} \maketitle

\begin{abstract} The bin packing  problem is to find the minimum
number of bins of size one to pack a list of items with sizes
$a_1,\ldots , a_n$ in $(0,1]$. Using uniform sampling, which selects
a random element from the input list each time, we develop a
randomized $O({n(\log n)(\log\log n)\over \sum_{i=1}^n a_i}+({1\over
\epsilon})^{O({1\over\epsilon})})$ time $(1+\epsilon)$-approximation
scheme for the bin packing problem. We show that every randomized
algorithm with uniform random sampling needs $\Omega({n\over
\sum_{i=1}^n a_i})$ time to give an $(1+\epsilon)$-approximation.
For each function $s(n): N\rightarrow N$, define $\sum(s(n))$ to be
the set of all bin packing problems with the sum of item sizes equal
to $s(n)$. For a constant $b\in (0,1)$, every problem in
$\sum(n^{b})$ has an $O(n^{1-b}(\log n)(\log\log n)+({1\over
\epsilon})^{O({1\over\epsilon})})$ time $(1+\epsilon)$-approximation
for an arbitrary constant $\epsilon$. On the other hand, there is no
$o(n^{1-b})$ time $(1+\epsilon)$-approximation scheme for the bin
packing problems in $\sum(n^{b})$ for some constant $\epsilon>0$. We
show that $\sum(n^{b})$ is NP-hard for every $b\in (0,1]$.
 This implies a dense sublinear
time hierarchy of approximation schemes for a class of NP-hard
problems, which are derived from the bin packing problem.
%
We also show  a randomized streaming approximation scheme for the
bin packing problem such that
  it needs only constant updating time and constant space,
  and outputs an $(1+\epsilon)$-approximation in $({1\over
\epsilon})^{O({1\over\epsilon})}$ time.
Let $S(\delta)$-bin packing be the class of bin packing problems
with each input item of size at least $\delta$. This research also
gives a natural example of NP-hard problem ($S(\delta)$-bin packing)
that has a constant time approximation scheme, and a constant time
and space sliding window streaming approximation scheme, where
$\delta$ is a positive constant.

\end{abstract}

\newpage

\section{Introduction}
 The bin packing  problem
is to find the minimum number of bins of size one to pack a list of
items  with sizes $a_1,\ldots , a_n$ in $(0,1]$. It is a classical
NP-hard problem and has been widely studied. The bin packing problem
has many applications in the engineering and information sciences.
Some approximation algorithm has been developed for bin packing
problem: for examples,
 the first fit, best fit, sum-of-squares, or
Gilmore-Gomory
cuts~\cite{ApplegateBuriolDillardJohnsonShor03,CsirikJohnsonKenyonOrlinShoreWeber00,CsirikJohnsonKenyonOrlinShoreWeber99,GilmoreJohnson61,GilmoreGomory63}.
The first linear time approximation scheme is shown
in~\cite{FernadezLueker81}. Recently, a sublinear time
${O}(\sqrt{n})$ with weighted sampling and a sublinear time
${O}(n^{1/3})$ with a combination of weighted and uniform samplings
were shown for bin packing problem~\cite{BatuBerenbrinkSohler09}.

We study the bin packing problem in randomized offline sublinear
time model, randomized streaming model, and randomized sliding
window streaming model. We also study the bin packing problem that
has input item sizes to be random numbers in $[0,1]$. Sublinear time
algorithms have been found for many computational problems, such as
checking polygon intersections
~\cite{ChazelleLiuMagen05},
estimating the cost of a minimum spanning
tree~\cite{ChazelleRubinfeldTrevisan05,CzumajErgun05,CzumajSohler04},
finding geometric separators~\cite{FuChen06},
and property testing~\cite{GoldreichGoldwasserRon96,GoldreichRon00},
etc.
Early research on streaming algorithms dealt with
simple statistics of the input data streams, such as the
median~\cite{MunroPaterson80}, the number of distinct
elements~\cite{FlajoletMartin85}, or frequency
moments~\cite{AlonMatiasSzegedy96}. Streaming algorithm is becoming
more and more important due to the development of internet, which
brings a lot of applications. There are many streaming algorithms
that have been proposed from the areas of computational theory,
database, and networking, etc.

 Due to the important role of
bin packing problem in the development of algorithm design and its
application in many other fields, it is essential to study the bin
packing problem in these natural
models. Our offline approximation scheme is based on the uniform
sampling, which selects a random element from the input list each
time. Our first approach is to approximate the bin packing problem
with a small number of samples under uniform sampling. We identify
that the complexity of approximation for the bin packing problem
inversely depends on the sum of the sizes of input items.

Using uniform sampling, we develop a randomized $O({n(\log
n)(\log\log n)\over \sum_{i=1}^n a_i}+({1\over
\epsilon})^{O({1\over\epsilon})})$ time $(1+\epsilon)$-approximation
scheme for the bin packing problem. We show that every randomized
algorithm with uniform random sampling needs $\Omega({n\over
\sum_{i=1}^n a_i})$ time to give an $(1+\epsilon)$-approximation.
Based on an adaptive random sampling method developed in this paper,
our algorithm automatically detects an approximation to the weights
of summation of the input items in time $O({n(\log n)(\log\log
n)\over \sum_{i=1}^n a_i})$ time, and then yields an
$(1+\epsilon)$-approximation.

For each function $s(n): N\rightarrow N$, define $\sum(s(n))$ to be
the set of all bin packing problems with the sum of item sizes equal
to $s(n)$. For a constant $b\in (0,1)$, every problem in
$\sum(n^{b})$ has an $O(n^{1-b}(\log n)(\log\log n)+({1\over
\epsilon})^{O({1\over\epsilon})})$ time $(1+\epsilon)$-approximation
for an arbitrary constant $\epsilon$. On the other hand, there is no
$o(n^{1-b})$ time $(1+\epsilon)$-approximation scheme for the bin
packing problems in $\sum(n^{b})$ for some constant $\epsilon>0$. We
show that $\sum(n^{b})$ is NP-hard for every $b\in (0,1]$.
 This implies a dense sublinear
time hierarchy of approximation schemes for a class of NP-hard
problems that are derived from bin packing problem.
%
We also show  a randomized single pass streaming approximation
scheme for the bin packing problem such that
  it needs only constant updating time and constant space,
  and outputs an $(1+\epsilon)$-approximation in $({1\over
\epsilon})^{O({1\over\epsilon})}$ time.
This research also gives an natural example of NP-hard problem that
has a constant time approximation scheme, and a constant time and
space sliding window  single pass  streaming approximation scheme.

 The
streaming algorithms in this paper for bin packing problem only
approximate the minimum number of bins to pack those input items. It
also gives a packing plan that allows an item position to be changed
at different moment. This has no contradiction with the existing
lower bound~\cite{Brown79,Liang80} that no approximation scheme
exists for online algorithm that does not change bins of  already
packed items.

A more general model of bin packing is studied in this paper. Given
a list of items in $(0,1]$, allocate  them  to several kinds of bins
with variant sizes and weights.
 We want to minimize the total costs $\sum_{i=1}^k u_iw_i$,
where $u_i$ is the number of bins of size $s_i$ and cost $w_i$.

In section~\ref{model-sec}, we give a description of computational
models used in this paper. A brief description of our methods are
also presented. In section~\ref{adaptive-sampling-sec}, we show an
adaptive random sampling method for the bin packing problem. In
section~\ref{rand-alg-sec}, we present randomized algorithms and
their lower bound for offline bin packing problem. In
section~\ref{streaming-sec}, we show a streaming approximation
scheme for bin packing problem. In
section~\ref{sliding-windows-sect}, we show a sliding window
streaming approximation scheme for bin packing problem with each
input item of size at least a positive constant $\delta$. The main
result of this paper is stated in
Theorem~\ref{random-approx-theorem}.


\section{Models of Computation and Overview of
Methods}\label{model-sec}

Algorithms for bin packing problem in this paper are under four
models, which are deterministic, randomized, streaming, and sliding
windows streaming models.

\begin{definition}\label{basic-bin-packing-def}\scrod
\begin{itemize}
\item
A {\it bin packing} is an allocation of the input items of sizes
$a_1,\ldots , a_n$ in $(0,1]$ to bins of size $1$.  We want to
minimize the total number of bins. We often use $Opt(L)$ to denote
the least number bins for packing items in $L$.
\item
Assume that $c$ and $\eta$ are constants in $(0,1)$, and $k$ is a
constant integer. There are $k$ kinds of bins of different sizes. If
$c\le s_i\le 1$, and $\eta\le w_i\le 1$ for all $i=1,2,\ldots , k$,
then we call the $k$ kinds of bins to be {\it $(c,\eta,k)$-related},
where $w_i$ and $s_i$ are the cost and size of the $i$-th kind of
bin, respectively.
\item
A {\it bin packing with $(c,\eta,k)$-related bins} is to allocate
the input items $a_1,\ldots , a_n$ in $(0,1]$ to
$(c,\eta,k)$-related bins. We want to minimize the total costs
$\sum_{i=1}^k u_iw_i$, where $u_i$ is the number of bins of cost
$w_i$. We often use $Opt_{c,\eta,k}(L)$ to denote the least cost for
packing items in $L$ with $(c,\eta,k)$-related bins. It is easy to
see $Opt(L)=Opt_{1,1,1}(L)$.
\item
For a positive constant $\delta$, a $S(\delta)$-bin packing  problem
is the bin packing  problem with all input items at least $\delta$.
\item
For a nondecreasing function $f(n):N\rightarrow N$, a
$\sum(f(n))$-bin packing  problem is the bin packing  problem with
all input items $a_1,\ldots ,a_n$ satisfying $\sum_{i=1}^n
a_i=f(n)$.
\end{itemize}
\end{definition}

{\bf Deterministic Model:}
 The bin packing problem under the deterministic model has been well
 studied. We give a generalized version of bin packing problem that allows multiple sizes of
 bins to pack them. It is called as bin packing with $(c,\eta,k)$ related bins in Definition~\ref{basic-bin-packing-def}. It is presented in
 Section~\ref{deterministic-section}.

{\bf Randomized Models:} Our main model of computation is based on
the uniform random sampling. We give the definitions for both
uniform and weighted random samplings below.

\begin{definition}
Assume that $a_1,\ldots , a_n$ is an input list of items in $(0,1]$
for a bin packing problem.
\begin{itemize}
\item
A {\it uniform sampling}~selects an element $a$ from the input list
with $\prob[a=a_i]={1\over n}$ for $i=1,\ldots , n$.
\item
A {\it weighted sampling}~selects an element $a$ from the input list
with $\prob[a=a_i]={a_i\over \sum_{i=1}^n a_i}$ for $i=1,\ldots ,
n$.
\end{itemize}
\end{definition}

We feel that the uniform sampling is more practical to implement
than weighted sampling. In this paper, our offline randomized
algorithms are based on uniform sampling. The weighted sampling was
used in~\cite{BatuBerenbrinkSohler09}.
The description of our offline algorithm with uniform random
sampling is given in Section~\ref{rand-alg-sec}.

{\bf Streaming Computation:} A data stream is an ordered sequence of
data items $p_1,p_2,\ldots , p_n$. Here, $n$ denotes the number of
data points in the stream. A {\it streaming algorithm}
  is an algorithm that computes some function over a data stream and
has the following properties: 1. The input data are accessed in the
sequential order of the data stream. 2. The order of the data items
in the stream is not controlled by the algorithm.
 Our algorithm for this model is
presented in Section~\ref{streaming-sec}.

{\bf Sliding Window Model:} In the sliding window streaming model,
there is a window size $n$ for the most recent $n$ items. The bin
packing problem for the sliding window streaming algorithm is to
pack the most recent $n$ items. Our algorithm for this model is
presented in Section~\ref{sliding-windows-sect}.

{\bf Bin Packing with Random Inputs:} We study the bin packing
problem such that the input is a series of
 sizes that are random numbers in $[0,1]$. It has a constant time
 approximation scheme and will be presented in
 Section~\ref{random-input-sec}.

\subsection{Overview of Our Method}\label{idea-overview-sec}

We develop algorithms for the bin packing problem under offline
uniform random sampling model, the streaming computation model, and
sliding window streaming model (only for $S(\delta)$-bin packing
with a positive constant $\delta$). The brief ideas are given below.

\subsubsection{Sublinear Time Algorithm for Offline Bin Packing}

Since the sum of input item sizes is not a part of input, it needs
$O(n)$ time to compute its exact value, and it's unlikely to be
approximated via one round random sampling in a sublinear time. We
first approximate the sum of sizes of items through a multi-phase
adaptive random sampling. Select a constant $\varphi$ to be the
threshold for large items. Select a small constant
$\gamma=O(\epsilon)$. All the items from the input are partitioned
into intervals $[\pi_1,\pi_0],(\pi_2,\pi_1]\ldots , (\pi_{i+1},
\pi_i],\ldots $ such that $\pi_0=1, \pi_1=\varphi$, and
$\pi_{i+1}=\pi_i/(1+\gamma)$ for $i=2,\ldots $. We approximate the
number of items in each interval $(\pi_{i+1}, \pi_i]$ via uniform
random sampling. Those intervals with very a small number of items
will be dropped. This does not affect much of the ratio of
approximation. One of worst cases is that all small items are of
size ${1\over n^2}$ and all large size items are of size $1$. In
this case, we need to sample $\Omega({n\over \sum_{a_i=1} 1})$
number of items to approximate the number of $1$s. This makes the
total time to be $\Omega({n\over \sum_{i=1}^n a_i})$. Packing the
items of large size is adapted the method
in~\cite{FernadezLueker81}, which uses a linear programming method
to pack the set of all large items, and fills small items into those
bins with large items to waste only a small piece of space for each
bin. Then the small items are put into bins that still have space
left after packing large items. When the sum of all item sizes is
$O(1)$, we need $O(n)$ time. Thus, the $O(n)$ time algorithm is a
part of our algorithm for the case $\sum_{i=1}^n a_i=O(1)$.


\subsubsection{Streaming Algorithm for Bin
Packing}\label{probability-subsec}

We apply the above approximation scheme to construct a  single pass
streaming algorithm for bin packing  problem. A crucial step is to
sample some random elements among those input items of size at least
$\delta$, which is set according to $\epsilon$. The weights of small
items are added to a variable $s_1$. After packing large items of
size at least $\delta$, we pack small items into those bins so that
each bin does not waste more than $\delta$ space while there is
small items unpacked.

\subsubsection{Sliding Window Streaming Algorithm for $S(\delta)$-Bin Packing}

Our sliding window  single pass streaming algorithm deals with the
bin packing problem that all input items are of size at least a
constant $\delta$. Let $n$ be the size of sliding window instead of
the total number of input items. Select a sufficiently large
constant $k$. There are $k$ sessions to approximate the bin packing.
After receiving every ${n\over k}$ items, a new session is started
to approximate the bin packing. The approximation ratio is
guaranteed via ignoring at most ${n\over k}$ items. As each item is
of large size at least $\delta$, ignoring ${n\over k}$ items only
affect a small ratio of approximation.


\subsubsection{Chernoff Bounds} The analysis of our randomized
algorithm often use  the well known Chernoff bounds, which are
described below. All proofs of this paper are self-contained except
the following famous theorems in probability theory and the
existence of a polynomial time algorithm for linear programming.

\begin{theorem}[\cite{MotwaniRaghavan00}]\label{chernoff-theorem}
Let $X_1,\ldots , X_n$ be $n$ independent random $0$-$1$ variables,
where $X_i$ takes $1$ with probability $p_i$. Let $X=\sum_{i=1}^n
X_i$, and $\mu=E[X]$. Then for any $\delta>0$,
\begin{enumerate}
\item $\Pr(X<(1-\delta)\mu)<e^{-{1\over 2}\mu\delta^2}$, and
\item
$\Pr(X>(1+\delta)\mu)<\left[{e^{\delta}\over
(1+\delta)^{(1+\delta)}}\right]^{\mu}$.
\end{enumerate}
\end{theorem}

We follow the proof of Theorem~\ref{chernoff-theorem} to make the
following versions (Theorem~\ref{ourchernoff2-theorem},
Theorem~\ref{chernoff3-theorem}, and
Corollary~\ref{chernoff-lemma-a}) of Chernoff bound for our
algorithm analysis.

\begin{theorem}\label{chernoff3-theorem}
Let $X_1,\ldots , X_n$ be $n$ independent random $0$-$1$ variables,
where $X_i$ takes $1$ with probability at least $p$ for $i=1,\ldots
, n$. Let $X=\sum_{i=1}^n X_i$, and $\mu=E[X]$. Then for any
$\delta>0$,
 $\Pr(X<(1-\delta)pn)<e^{-{1\over 2}\delta^2 pn}$.
\end{theorem}

\begin{theorem}\label{ourchernoff2-theorem}
Let $X_1,\ldots , X_n$ be $n$ independent random $0$-$1$ variables,
where $X_i$ takes $1$ with probability at most $p$ for $i=1,\ldots ,
n$. Let $X=\sum_{i=1}^n X_i$. Then for any $\delta>0$,
$\Pr(X>(1+\delta)pn)<\left[{e^{\delta}\over
(1+\delta)^{(1+\delta)}}\right]^{pn}$.
\end{theorem}

Define $g_1(\delta)=e^{-{1\over 2}\delta^2}$ and
$g_2(\delta)={e^{\delta}\over (1+\delta)^{(1+\delta)}}$. Define
$g(\delta)=\max(g_1(\delta),g_2(\delta))$. We note that
$g_1(\delta)$ and $g_2(\delta)$ are always strictly less than $1$
for all $\delta>0$. It is trivial for $g_1(\delta)$. For
$g_2(\delta)$, this can be verified by checking that the function
$f(x)=(1+x)\ln (1+x)-x$ is increasing and $f(0)=0$. This is because
$f'(x)=\ln (1+x)$ which is strictly greater than $0$ for all $x>0$.

\begin{corollary}[\cite{LiMaWang99}]\label{chernoff-lemma-a}
Let $X_1,\ldots , X_n$ be $n$ independent random $0$-$1$ variables
and $X=\sum_{i=1}^n X_i$.

 i. If $X_i$ takes $1$ with probability at most $p$ for $i=1,\ldots ,
n$, then for any ${1\over 3}>\epsilon>0$, $\Pr(X>pn+\epsilon
n)<e^{-{1\over 3}n\epsilon^2}$.

ii. If  $X_i$ takes $1$ with probability at least $p$ for
$i=1,\ldots , n$, then for any $\epsilon>0$, $\Pr(X<pn-\epsilon
n)<e^{-{1\over 2}n\epsilon^2}$.
\end{corollary}

A well known fact in probability theory is the inequality
$$\Pr(E_1\cup E_2 \ldots \cup E_m)\le
\Pr(E_1)+\Pr(E_2)+\ldots+\Pr(E_m),$$ where $E_1,E_2,\ldots, E_m$ are
$m$ events that may not be independent. In the analysis of our
randomized algorithm, there are multiple events such that the
failure from any of them may fail the entire algorithm. We often
characterize the failure probability of each of those events, and
use the above inequality to show that the whole algorithm has a
small chance to fail after showing that each of them has a small
chance to fail.

\section{Adaptive Random Sampling for Bin Packing}\label{adaptive-sampling-sec}
In this section, we develop an adaptive random sampling method to
get the rough information for a list of items for the bin packing
problem. We show a randomized  algorithm to approximate the sum of
the sizes of input items in $O(({n\over \sum_{i=1}^n a_i})(\log
n)\log\log n))$ time. This is the core step of our randomized
algorithm, and is also or main technical contribution.

\begin{definition}\label{partition-def}\scrod
\begin{itemize}
\item
For each interval $I$ and a list of items $S$, define $C(I,S)$ to be
the number of items of $S$ in $I$.

\item
For $\varphi,\delta$, and $\gamma$ in $(0,1)$, a {\it
$(\varphi,\delta,\gamma)$-partition} for $(0,1]$ divides the
interval $(0,1]$ into intervals $I_1=[\pi_1, \pi_0], I_2=(\pi_2,
\pi_1], I_3=(\pi_3,\pi_2],\ldots ,I_k=(0, \pi_{k-1}]$
 such that $\pi_0=1, \pi_1=\varphi,
\pi_i=\pi_{i-1}(1-\delta)$ for $i=2,\ldots , k-1$, and $\pi_{k-1}$
is the first element $\pi_{k-1}\le {\gamma\over n^2}$.

\item
For a set $A$, $|A|$ is the number of elements in $A$. For a list
$S$ of items, $|S|$ is the number of items in $S$.
\end{itemize}
\end{definition}

\begin{lemma}
For parameters $\varphi,\delta$, and $\gamma$ in $(0,1)$, a
$(\varphi,\delta,\gamma)$-partition for $(0,1]$ has the number of
intervals $k\le {2\log n\over \gamma \theta}$.
\end{lemma}

\begin{proof}
The number of intervals $k$ is the least integer with
$\delta(1-\delta)^k\le (1-\delta)^k\le {\gamma\over n^2}$. We have
$k\le {\log {n^2\over \gamma}\over \log (1-\delta)}\le {2\log n\over
\gamma\delta}$.
\end{proof}

We need to approximate the number of large items, the total sum of
the sizes of items, and the total sum of the sizes of small items.
For a $(\varphi,\delta,\gamma)$-partition $I_1\cup I_2\ldots  \cup
I_k$ for $(0,1]$, Algorithm Approximate-Intervals$(.)$ below gives
the estimation for the number of items in each $I_j$ if interval
$I_j$ has a number items to be large enough. Otherwise, those items
in $I_j$ can be ignored without affecting much of the approximation
ratio. We have an adaptive way to do random samplings in a series of
phases. Phase $t+1$ doubles the number of random samples of phase
$t$ ($m_{t+1}=2m_t$). For each phase, if an interval $I_j$ shows
sufficient number of items from the random samples, the number of
items $C(I_j,S)$ in $I_j$ can be sufficiently approximated by
$\hat{C}(I_j,S)$. Thus, $\hat{C}(I_j,S)\pi_j$ also gives an
approximation for the sum of the sizes of items in $I_j$. The sum
$app_w=\sum_{I_j}\hat{C}(I_j,S)\pi_j$ for those intervals $I_j$ with
large number of samples gives an approximation for the total sum
$\sum_{i=1}^na_i$ of items in the input list. Let $m_t$ denote the
number of random samples in phase $t$. In the early stages, $app_w$
is much smaller than ${n\over m_t}$. Eventually, $app_w$ will
surpass ${n\over m_t}$. This happens when $m_t$ is more than
${n\over \sum_{i=1}^n a_i}$ and $app_w$ is close to the sum
$\sum_{i=1}^n a_i$ of all items from the input list. This indicates
that the number of random samples is sufficient for approximation
algorithm. For those intervals with small number of samples, their
items only need small fraction of bins to be packed. This process is
terminated when ignoring all those intervals with none or small
number of samples does not affect much of the accuracy of
approximation. The algorithm gives up the process of random sampling
when $m_t$ surpasses $n$, and switches to use a deterministic way to
access the input list, which happens when the total sum of the sizes
of input items is $O(1)$. The lengthy analysis is caused by the
multi-phases adaptive random samplings. We show two examples below.

{\bf Example 1:} The  input is a list of items such that there are
three items of size $1$, and the rest $n-3$ items are of size $0.1$
for a large integer $n$.  Assume that $\epsilon$ is a positive
constant to control the accuracy of approximation. After sampling a
constant ${100\over \epsilon}$ number of items, we observe all
samples equal to $0.1$ (with high probability). Thus, there are less
than ${\epsilon n\over 20}$ items of size other than $0.1$ with high
probability by Chernoff bounds. We derive the approximate sum of
total item sizes is $0.1n$, and output ${0.1(1+\epsilon)n\over 0.9}$
for the number bins for packing the input items, where the
denominator $0.9$ is based on the consideration that some bins for
packing items of size $0.1$ may waste up to $0.1$ space. Although,
there are small number of items of size $1$, just ignoring those
items of size $1$ loses only a small accuracy of approximation.
Therefore, the random sampling stops after sampling only $O({1\over
\epsilon})$ items. We output an $(1+\epsilon)$-approximation for the
bin packing problem.

{\bf Example 2:} The  input is a list of items such that there are
three items are of size $1$, and the rest $n-3$ items are of size
${1\over n^2}$ for a large integer $n$.  The number of random
samples is doubled from one phase to next phase. After sampling
$n^{0.9}$ items, in which there is no large items of size $1$ with
high probability, we still feel that those items of large size will
greatly affect the total number bins. We have to continue use more
random samples. Eventually, the number of random samples $m_t$ is
more than $n$. Thus, we switch to use a deterministic $O(n)$ time
algorithm to compute the number of large items, the total sum of the
sizes of items, and the total sum of the sizes of  small items.

\vskip 10pt

 {\bf Algorithm Approximate-Intervals$(\varphi, \delta, \gamma,\theta, \alpha, P, n, S)$}

Input: a parameter $\varphi\in (0,1)$, a small parameter $\theta\in
(0,1)$,  a failure probability upper bound $\alpha$, a
$(\varphi,\delta,\gamma)$ partition $P=I_1\cup\ldots \cup I_k$ for
$(0,1]$ with $\delta,\gamma\in (0,1)$, an integer $n$, a list $S$ of
$n$ items $a_1,\ldots , a_n$ in $(0,1]$. Parameters $\varphi,
\delta, \gamma,\theta,$ and $\alpha$ do not depend on the number of
items $n$.


Steps:

\begin{enumerate}[1.]

\item
Phase $0$:

\item\label{first-alpha-setting}
\qquad Let $z:=\xi_0\log \log n$, where $\xi_0$ is a parameter such
that $8(k+1)(\log n) g(\theta)^{z/2}<\alpha$ for all large $n$.

\item\label{constant-setting-in-Approximate-Intervals}
\qquad Let parameters $c_0:={1\over 100},  c_2:={1\over
3(1+\delta)c_0}, c_3:={\delta^4\over 2(1+\delta)}, c_4:={8\over
(1-\theta)(1-\delta)\varphi c_0}$, and $c_5:={12\xi_0\over
(1-\theta) c_2c_3}$.

\item
\qquad Let $m_0:=z$.

\item
End of Phase $0$.

\item
Phase $t$:


\item\label{loop-m-start}
\qquad Let $m_t:=2m_{t-1}$.

\item
\qquad Sample $m_t$ random items $a_{i_1},\ldots , a_{i_{m_t}}$ from
the input list $S$.

\item
\qquad Let $d_j:=|\{j: a_{i_j}\in I_j\ \ and \ 1\le j\le m_t\}|$ for
$j=1,2,\ldots , k$.

\item\label{I-j-loop}
\qquad For each $I_j$,

\item
\qquad\qquad if $d_j\ge z$,

\item\label{assign-hat-C}
\qquad\qquad then let $\hat{C}(I_j,S):={n\over m_t}d_j$ to
approximate $C(I_j,S)$.

\item\label{I-j-loop-end}
\qquad\qquad else let $\hat{C}(I_j,S):=0$.

\item\label{loop-m-end}
\qquad Let $app_w:=\sum_{d_j\ge z}\hat{C}(I_j,S)\pi_j$ to
approximate $\sum_{i=1}^n a_n$.

\item\label{until-condition}
\qquad If $app_w\le {c_5n\log\log n\over c_0m_t}$ and $m_t< n$ then
enter Phase $t+1$.

\item
\qquad else

\item
\qquad\qquad If $m_t<n$

\item
\qquad\qquad then let $app_w':=\sum_{d_j\ge z\ and \
j>1}\hat{C}(I_j,S)\pi_j$ to approximate $\sum_{a_i<\delta,1\le i\le
n}a_i$.

\item
\qquad\qquad else let $app_w:=\sum_{i=1}^na_i$ and
$app_w':=\sum_{a_i<\varphi} a_i$.

\item
\qquad\qquad Output $app_w$, $app_w'$ and $\hat{C}(I_1,S)$ (the
approximate number of items  of size at least $\varphi$).

\item
End of Phase $t$.
\end{enumerate}

{\bf End of Algorithm}

\vskip 10pt

Lemma~\ref{app-sum-lemma} uses several parameters $\varphi, \delta,
\gamma,\alpha$ and $\theta$ that will be determined by the
approximation ratio for the the bin packing problem. If the
approximation ratio is fixed, they all  become constants.

\begin{lemma}\label{app-sum-lemma} Assume that $\varphi, \delta, \gamma,\alpha$ and $\theta$ are parameters in $(0,1)$, and those parameters do not depend on the number of
items $n$.. Then there exists a randomized  algorithm  described in
Approximate-Intervals(.) such that given a list $S$ of items of size
$a_1,\ldots , a_n$ in the range $(0,1]$ and a
$(\varphi,\delta,\gamma)$-partition for $(0,1]$, with probability at
most $\alpha$, at least one of the following statements is false
after executing the algorithm:
\begin{enumerate}[1.]
\item\label{item1-app-sum-lemma}
 For each $I_j$ with $\hat{C}(I_j,S)>0$, ${C(I_j,S)
(1-\theta)}\le \hat{C}(I_j,S)\le {C(I_j,S)(1+\theta)}$;
\item\label{item2-app-sum-lemma}
$\sum_{a_i\in I_j\
and\ \hat{C}(I_j,S)=0}a_i\le {\delta^3\over 2}(\sum_{i=1}^n
a_i)+{\gamma\over n}$;
\item\label{item3-app-sum-lemma}
$(1-\theta)(1-\delta)\varphi({\sum_{i=1}^na_i\over 2}-{2\gamma\over
n})\le app_w\le (1+\theta)(\sum_{i=1}^na_i)$;
\item\label{item3b-app-sum-lemma}
If $\sum_{i=1}^na_i\ge 4$, then ${1\over
4}(1-\theta)(1-\delta)\varphi(\sum_{i=1}^na_i)\le app_w\le
(1+\theta)(\sum_{i=1}^na_i)$; and
\item\label{item4-app-sum-lemma}
It runs in $O({1\over (1-\theta)\delta^4\log g(\theta)}\min({n\over
\sum_{i=1}^n a_i},n)(\log n)\log\log n)$ time. In particular, the
complexity of the algorithm is $O(\min({n\over \sum_{i=1}^n
a_i},n)(\log n)\log\log n)$ if $\varphi, \delta, \gamma,\alpha$ and
$\theta$ are constants in $(0,1)$.
\end{enumerate}
\end{lemma}

Lemma~\ref{app-sum-lemma} implies that with probability at least
$1-\alpha$, all statements~\ref{item1-app-sum-lemma}
to~\ref{item4-app-sum-lemma} are true. Due to the technical reason
described at the end of section~\ref{probability-subsec}, we
estimate the failure probability instead of the success probability.

\begin{proof}
Let $\xi_0, c_0,c_2,c_3,c_4$, and $c_5$ be parameters defined as
those in the algorithm Approximate-Intervals$(.)$.
We use the uniform random sampling to approximate the  number of
items in each interval $I_j$ in the
$(\varphi,\delta,\gamma)$-partition.




{\bf Claim~\ref{app-sum-lemma}.1.} Let $Q_1$ be the probability that
the following statement is false:

(i) For each interval $I_j$ with $d_j\ge z$, $(1-\theta)C(I_j,S)\le
\hat{C}(I_j,S)\le (1+\theta)C(I_j,S)$.

Then for each phase in the algorithm, $Q_1\le (k+1)\cdot
g(\theta)^{z\over 2}$.

\begin{proof}
Let $p_j={C(I_j,S)\over n}$. An element of $S$ in $I_j$ is sampled
(by an uniform sampling) with probability $p_j$. Let $p'={z\over
2m_t}$.
 For each interval $I_j$ with $d_j\ge z$, we discuss two
cases.

\begin{itemize}
\item
Case 1. $p'\ge p_j$.

In this case, $d_j\ge z\ge 2p' m_t\ge 2p_j m_t$. Note that $d_j$ is
the number of elements in interval $I_j$ among $m_t$ random samples
$a_{i_1},\ldots , a_{i_{m_t}}$ from $S$. By
Theorem~\ref{ourchernoff2-theorem} (with $\theta=1$), with
probability at most $P_1=g_2(1)^{p'm_t}\le g_2(1)^{z/2}\le
g(1)^{z/2}$, there are at least $2p_j m_t$ samples are in from
interval $I_j$.

\item
Case 2. $p'< p_j$.

 By
Theorem~\ref{ourchernoff2-theorem}, we have $\prob[d_j>
(1+\theta)p_jm_t]\le g_2(\theta)^{p_jm_t}\le g_2(\theta)^{p'm_t}\le
g_2(\theta)^{z\over 2}\le g(\theta)^{z\over 2}$.


By Theorem~\ref{chernoff3-theorem}, we have $\prob[d_j\le
(1-\theta)p_jm_t]\le g_1(\theta)^ {p_jm_t}\le
g_1(\theta)^{p'm_t}=g_1(\theta)^{z\over 2}\le g(\theta)^{z\over 2}$.

For each interval $I_j$ with $d_j\ge z$ and $(1-\theta)p_jm_t\le
d_j\le (1+\theta)p_jm_t$, we have $(1-\theta)C(I_j,S)\le
\hat{C}(I_j,S)\le (1+\theta)C(I_j,S)$ by line~\ref{assign-hat-C} in
Approximate-Intervals(.).

There are $k=(\log n)$ intervals $I_1,\ldots , I_k$. Therefore, with
probability at most $P_2=k\cdot g(\theta)^{z\over 2}$,
 the following is false: For each interval $I_j$
with $d_j\ge z$, $(1-\theta)C(I_j,S)\le \hat{C}(I_j,S)\le
(1+\theta)C(I_j,S)$.
\end{itemize}

 By the
analysis of Case 1 and Case 2, we have  $Q_1\le P_1+P_2\le
(k+1)\cdot g(\theta)^{z\over 2}$. Thus, the claim has been proven.
\end{proof}

{\bf Claim \ref{app-sum-lemma}.2.} Assume that $m_t\ge{c_2
c_5n\log\log n\over \sum_{i=1}^n a_i}$. Then right after executing
Phase $t$ in Approximate-Intervals$(.)$, with probability at most
$Q_2=2kg(\theta)^{\xi_0\log\log n}$, the following statement is
false:

(ii) For each interval $I_j$ with $C(I_j, S)\ge c_3\sum_{i=1}^n
a_i$, A). $(1-\theta)C(I_j,S)\le \hat{C}(I_j,S)\le
(1+\theta)C(I_j,S)$; and B). $d_j\ge z$.

\begin{proof}
Assume that $m_t\ge{c_2 c_5n\log\log n\over \sum_{i=1}^n a_i}$.
Consider each interval $I_j$ with $C(I_j, S)\ge c_3\sum_{i=1}^n
a_i$. We have that $p_j={C(I_j,S)\over n}\ge {c_3\sum_{i=1}^n
a_i\over n}$. An element of $S$ in $I_j$ is sampled with probability
$p_j$.  By Theorem~\ref{ourchernoff2-theorem} and
Theorem~\ref{chernoff3-theorem}, we have
\begin{eqnarray}
\prob[d_j<(1-\theta)p_jm_t]\le g_1(\theta)^{p_jm_t}\le
g_1(\theta)^{c_2c_3c_5\log\log n}\le g(\theta)^{\xi_0\log\log n}.\\
\prob[d_j>(1+\theta)p_jm_t]\le g_2(\theta)^{p_jm_t}\le
g_2(\theta)^{c_2c_3c_5\log\log n}\le g(\theta)^{\xi_0\log\log n}.
\end{eqnarray}

Therefore, with probability at most $2kg(\theta)^{\xi_0\log\log n}$,
the following statement is false:

For each interval $I_j$ with $C(I_j, S)\ge c_3\sum_{i=1}^n a_i$,
$(1-\theta)C(I_j,S)\le \hat{C}(I_j,S)\le (1+\theta)C(I_j,S)$.

If $d_j\ge (1-\theta)p_j m_t$, then we have
\begin{eqnarray*}
d_j&\ge&  (1-\theta){C(I_j, S)\over n} m_t\\
&\ge& (1-\theta){(c_3\sum_{i=1}^n a_i)\over n}\cdot {c_2
c_5n\log\log n\over \sum_{i=1}^n a_i}\\
&=&(1-\theta)c_2c_3c_5\log\log n\\
&\ge& \xi_0\log\log n= z.
\end{eqnarray*}
\end{proof}


{\bf Claim \ref{app-sum-lemma}.3.} The total sum of the sizes of
items in those $I_j$s with $C(I_j, S)< c_3\sum_{i=1}^n a_i$ is at
most ${ \delta^3\over 2}(\sum_{i=1}^n a_i)+{\gamma\over n}$.

\begin{proof} By definition~\ref{partition-def}, we have
$a_j=\varphi(1-\delta)^{j-1}$ for $j=1,\ldots ,k-1$. We have that
\begin{itemize}
\item
the sum of sizes of items in $I_k$ is at most $n{\gamma\over
n^2}={\gamma\over n}$,
\item
for each interval $I_j$ with $C(I_j, S)< c_3\sum_{i=1}^n a_i$, the
sum of sizes of items in $I_j$ is at most $(c_3\sum_{i=1}^n
a_i)a_{j-1}\le (c_3\sum_{i=1}^n a_i)\varphi(1-\delta)^{j-2}$ for
$j\in (1,k)$, and
\item
the sum of sizes in $I_1$ is at most $c_3\sum_{i=1}^n a_i$.
\end{itemize}
The total sum of the sizes of items in those $I_j$s with $C(I_j, S)<
c_3\sum_{i=1}^n a_i$  is at most $(c_3\sum_{i=1}^n
a_i)+\sum_{j=2}^k(c_3\sum_{i=1}^n
a_i)\varphi(1-\delta)^{j-2})+n\cdot {r\over n^2}\le (c_3\sum_{i=1}^n
a_i)+{c_3\varphi\over \delta}(\sum_{i=1}^n a_i)+{\gamma\over n}\le {
\delta^3\over 2}(\sum_{i=1}^n a_i)+{\gamma\over n}$.
\end{proof}


{\bf Claim \ref{app-sum-lemma}.4.} Assume that at the end of phase
$t$, for each $I_j$ with $\hat{C}(I_j,S)>0$, ${C(I_j,S)
(1-\theta)}\le \hat{C}(I_j,S)\le {C(I_j,S)(1+\theta)}$; and $d_j\ge
z$ if $C(I_j,S)\ge c_3\sum_{i=1}^na_i$. Then
$(1-\theta)(1-\delta)\varphi({\sum_{i=1}^na_i\over 2}-{2\gamma\over
n})\le app_w\le (1+\theta)(\sum_{i=1}^na_i)$ at the end of phase
$t$.

\begin{proof} By the assumption of the claim, we have $app_w=\sum_{d_j\ge z}\hat{C}(I_j,S)\pi_j\le
(1+\theta)\sum_{i=1}^n a_i$. For each interval $I_j$ with $j\not=k$
and $j>1$, we have $C(I_j,S)\pi_j\ge (1-\delta)\sum_{a_i\in I_j}a_i$
by the definition of $(\varphi,\delta,\gamma)$-partition. It is easy
to see that $C(I_1,S)\pi_1\ge \varphi\sum_{a_i\in I_1}a_i$ by the
definition of $(\varphi,\delta,\gamma)$-partition. Thus,
\begin{eqnarray}
C(I_j,S)\pi_j\ge (1-\delta)\varphi\sum_{a_i\in I_j}a_i\ \ \
\mbox{for \ \ $j\not=k$.}\label{lower-bound-one-interval}
\end{eqnarray}
 We have the following inequalities:
\begin{eqnarray*}
app_w&=&\sum_{d_j\ge z}\hat{C}(I_j,S)\pi_j\ \ \ \mbox{(by\ line~\ref{until-condition}\ in\ Approximate-Intervals(.))}\\
&\ge&(1-\theta)\sum_{d_j\ge z}C(I_j,S)\pi_j\\
&\ge&(1-\theta)\sum_{d_j\ge z, j\not=k}C(I_j,S)\pi_j\\
&\ge&(1-\theta)(1-\delta)\varphi\sum_{d_j\ge z,j\not=k}\left(\sum_{a_i\in I_j}a_i\right)\ \ \ \ \mbox{(by\ inequality~(\ref{lower-bound-one-interval}))}\\
&\ge&(1-\theta)(1-\delta)\varphi(\sum_{i=1}^na_i-\sum_{d_j<z}\sum_{a_i\in I_j} a_i-\sum_{a_i\in I_k}a_i)\\
&\ge&(1-\theta)(1-\delta)\varphi(\sum_{i=1}^na_i-({ \delta^3\over 2}(\sum_{i=1}^n a_i)+{\gamma\over n})-n\cdot {\gamma\over n^2})\ \ \ \mbox{(by\ Claim\ \ref{app-sum-lemma}.3)}\\
&\ge&(1-\theta)(1-\delta)\varphi({\sum_{i=1}^na_i\over
2}-{2\gamma\over n}).
\end{eqnarray*}
\end{proof}


{\bf Claim \ref{app-sum-lemma}.5.} With probability at most
$Q_5=(k+1)\cdot (\log n) g(\theta)^{z\over 2}$,  the following facts
are not all true:
\begin{enumerate}[A.]
\item\label{clm5-a}
For each phase $t$ with $m_t<{2c_2 c_5n\log\log n\over \sum_{i=1}^n
a_i}$,  the condition $app_w\le {c_5n\log\log n\over c_0m_t}$ in
line~\ref{until-condition} of the algorithm is true.
\item\label{clm5-b}
If $\sum_{i=1}^na_i\ge 4$, then the algorithm stops before $m_t>
{2c_4 c_5n\log\log n\over \sum_{i=1}^n a_i}$.

\item\label{clm5-c}
If $\sum_{i=1}^na_i\le 4$, then it stops before or at phase $t$ in
which the condition $m_t\ge n$ first becomes true.
\end{enumerate}

\begin{proof}
By Claim \ref{app-sum-lemma}.1, with probability at most $(k+1)\cdot
g(\theta)^{z\over 2}$, the statement i of Claim
\ref{app-sum-lemma}.1 is false for a fixed $m$. The number of phases
is at most $\log n$ since $m_t$ is double at each phase. With
probability $(k+1)\cdot (\log n)\cdot g(\theta)^{z\over 2}$, the
statement i of Claim \ref{app-sum-lemma}.1 is false for each phase
$t$ with $m_t\le n$. Assume that statement i of Claim
\ref{app-sum-lemma}.1 is true for all phases $t$ with $m_t\le n$.

Statement~\ref{clm5-a}.  Assume that $m_t<{2c_2 c_5n\log\log n\over
\sum_{i=1}^n a_i}$. We have ${n\over m_t}>{n\over {2c_2 c_5
n\log\log n\over \sum_{i=1}^n a_i}}={\sum_{i=1}^n a_i\over
2c_2c_5\log\log n}$. Therefore, $\sum_{i=1}^n a_i<({n\over
m_t})2c_2c_5\log\log n={2c_2c_5n\log\log n\over m_t}$.
 By
Claim \ref{app-sum-lemma}.4, $app_w\le (1+\theta)\sum_{i=1}^n a_i$.
 Since $(1+\theta)<{1\over 2c_2c_0}$ (by line~\ref{constant-setting-in-Approximate-Intervals} in Approximate-Intervals(.)), we have
\begin{eqnarray*}
 app_w\le (1+\theta)\sum_{i=1}^n a_i\le {1\over 2c_2c_0}\sum_{i=1}^n a_i
 < {1\over 2c_2c_0}\cdot {2c_2c_5n\log\log n\over m_t}= {c_5n\log\log n\over c_0m_t}.
\end{eqnarray*}
Statement~\ref{clm5-b}. The variable $m_t$ is doubled in each new
phase.
Assume that the algorithm enters phase $t$ with ${c_4 c_5n\log\log
n\over \sum_{i=1}^n a_i}\le m_t\le {2c_4 c_5n\log\log n\over
\sum_{i=1}^n a_i}$.  We have ${n\over m_t}\le {n\over {c_4c_5
n\log\log n\over \sum_{i=1}^n a_i}}={\sum_{i=1}^n a_i\over
c_4c_5\log\log n}$. Since $\sum_{i=1}^na_i\ge 4$,
$({\sum_{i=1}^na_i\over 2}-{\gamma\over n})\ge {\sum_{i=1}^na_i\over
4}$. By Claim \ref{app-sum-lemma}.4,
 $app_w$
is at least
 ${(1-\theta)(1-\delta)\varphi\over 4}\sum_{i=1}^n a_i$.
 Since ${(1-\theta)(1-\delta)\varphi\over 4}>{1\over c_0c_4}$, we have $app_w> {c_5n\log\log n\over
 c_0m}$, which makes the condition at
line~\ref{until-condition} in Approximate-Intervals(.) be false.
Thus, the algorithm stops at some stage $t$ with $m_t\le {2c_4
c_5n\log\log n\over \sum_{i=1}^n a_i}$ by the setting at
line~\ref{until-condition} in Approximate-Intervals(.).

 Statement~\ref{clm5-c}. It follows from statement A and the
setting in line~\ref{until-condition} of the algorithm.
\end{proof}

{\bf Claim \ref{app-sum-lemma}.6.} The complexity of the algorithm
is $O({1\over (1-\theta)\delta^4\log g(\theta)}\min({n\over
\sum_{i=1}^n a_i},n)(\log n)\log\log n)$. In particular, the
complexity is $O(\min({n\over \sum_{i=1}^n a_i},n)(\log n)\log\log
n)$ if $\varphi, \delta, \gamma,\alpha$ and $\theta$ are constants
in $(0,1)$.

\begin{proof} By the setting in line~\ref{constant-setting-in-Approximate-Intervals} in
Approximate-Intervals(.), we have
\begin{eqnarray*}
c_2c_5&=&{1\over 3(1+\delta)c_0}\cdot {12\xi_0\over (1-\theta)c_2c_3}\\
&=&{4\xi_0\over (1+\delta)\cdot c_0\cdot (1-\theta)\cdot {1\over 3(1+\delta)c_0}\cdot {\delta^4\over 2(1+\delta)}}\\
&=&{24\xi_0(1+\delta)\over (1-\theta)\delta^4}.
\end{eqnarray*}

In order to satisfy the condition
 $8(k+1)(\log n) g(\theta)^{z/2}<\alpha$ for all large $n$ at line~\ref{first-alpha-setting} in Approximate-Intervals(.), we can let $\xi_0={8\over \log g(\theta)}$.

 Since $m_t$ is doubled every phase, the total number of phases is at most $\log n$. The computational time complexity in statement~\ref{item4-app-sum-lemma}
   of the algorithm follows from
Claim \ref{app-sum-lemma}.5.
\end{proof}

As $m_t$ is doubled each new phase in Approximate-Intervals$(.)$,
the number of phases is at most $\log n$. With probability at most
$(\log n)(Q_1+Q_2)+Q_5\le \alpha$ (by line~\ref{first-alpha-setting}
in Approximate-Intervals$(.)$), at least one of the statements (i)
in Claim \ref{app-sum-lemma}.1, (ii) in Claim \ref{app-sum-lemma}.2,
A, B, C in Claim \ref{app-sum-lemma}.5 is false.


Assume that the statements (i) in Claim \ref{app-sum-lemma}.1, (ii)
in Claim \ref{app-sum-lemma}.2, A, B, and C in Claim
\ref{app-sum-lemma}.5 are all true.

For an interval $I_j$, $\hat{C}(I_j,S)>0$ if and only if $d_j\ge z$
by lines~\ref{I-j-loop} to \ref{I-j-loop-end} in
Approximate-Intervals(.). Therefore, statement 1 of the lemma
follows from Claim \ref{app-sum-lemma}.1.

If Approximate-Intervals(.) stops at $m_t<n$, then $m_t\ge{2c_2
c_5n\log\log n\over \sum_{i=1}^n a_i}$ by statement A in Claim
\ref{app-sum-lemma}.5.
 For each interval $I_j$ with $C(I_j, S)\ge c_3\sum_{i=1}^n
a_i$, we have $d_j\ge z$, which implies $\hat{C}(I_j,S)>0$.
 Statement~\ref{item2-app-sum-lemma} of Lemma~\ref{app-sum-lemma}
 follows from Claim \ref{app-sum-lemma}.3 and statement (ii) of Claim \ref{app-sum-lemma}.2.

Statement~\ref{item3-app-sum-lemma} follows from Claim
\ref{app-sum-lemma}.4. The condition of
Statement~\ref{item3b-app-sum-lemma} implies $n\ge 4$.
 Statement~\ref{item3b-app-sum-lemma} follows from
 Statement~\ref{item3-app-sum-lemma}.
Statement~\ref{item4-app-sum-lemma} for the running time follows
from Claim \ref{app-sum-lemma}.6.

 Thus,
with probability at  most $\alpha$, at least one of the
statements~\ref{item1-app-sum-lemma} to \ref{item4-app-sum-lemma} is
false.
\end{proof}

\section{Main Results}

We list the main results that we achieve in this paper. The  proof
of Theorem~\ref{random-approx-theorem} is shown in
Section~\ref{full-sublinear-section}.



\begin{theorem}[Main]\label{random-approx-theorem}
Approximate-Bin-Packing(.) is a randomized approximation scheme for
the bin packing problem such that given an arbitrary $\tau \in
(0,1)$ and a list of items $S=a_1,\ldots , a_n$ in $(0,1]$ for the
bin packing problem,  it gives an approximation $app$ with
$Opt(S)\le app\le (1+\tau  )Opt(S)+1$ in $O({n(\log n)(\log\log
n)\over \sum_{i=1} a_i}+({1\over \tau  })^{O({1\over\tau  })})$ time
with probability at least ${3\over 4}$.
\end{theorem}


We show a lower bound for those bin packing problems with bounded
sum of sizes $\sum_{i=1}^n a_i$. The lower bound always matches the
upper bound.

\begin{theorem}\label{strong-lower-bound-theorem}  Assume $f(n)$ is
a nondecreasing unbounded function from $N$ to $N$ with $f(n)=o(n)$.
Every randomized $(2-\epsilon)$ approximation algorithm for bin
packing problems in $\sum(f(n))$ needs $\Omega({n\over f(n)})$ time,
where $\epsilon$ is an arbitrary small constant in $(0,1)$.
\end{theorem}

\begin{proof} 
Since $f(n)$ is unbounded, assume $n$ is large enough such
that
\begin{eqnarray}
(f(n)+2)(2-\epsilon)<2(f(n)-2).\label{f(n)-bound-ineqn}
\end{eqnarray}
We design two input list of items.

The first list contains $m=2(f(n)-2))$ elements of size ${1\over
2}+\delta$, where $\delta={1\over 2(f(n)-2)}$. The rest $n-m$ items
are of the same size $\gamma={1\over n-m}=o(1)$. We have $m({1\over
2}+\delta)+(n-m)\gamma=2(f(n)-2)({1\over 2}+{1\over
2(f(n)-2)})+1=f(n)$. Therefore, the first list is a bin packing
problem is in $\sum(f(n))$.

 The second list contains $n-f(n)$ elements of size $\gamma$ and the rest $f(n)$ items are
of size equal to  $1-\tau$, where $\tau={(n-f(n))\gamma\over
f(n)}=o(1)$. We have $f(n)(1-\tau)+(n-f(n))\gamma=f(n)$. The second
list is also a bin packing problem is in $\sum(f(n))$.

Both $\gamma$ and $\tau$ are small. Packing the first list needs at
least $2(f(n)-2)$ bins. Packing the second list only needs at most
$f(n)+2$ bins since two bins of size one is enough to pack those
items of size $\tau$.

Assume that an algorithm only has computational time $o({n\over
f(n)})$ for computing $(2-\epsilon)$-approximation for bin packing
problems in $\sum(f(n))$. The algorithm has an $o(1)$ probability to
access at least one item of size at least ${1\over 2}$ in both
lists. Therefore, the two lists have the same output for
approximation by the same randomized algorithm. For the second list,
the output for the number of bins should be at most
$(f(n)+2)(2-\epsilon)$. By inequality~(\ref{f(n)-bound-ineqn}), it
is impossible to pack the first list items. This brings a
contradiction.
\end{proof}

\begin{corollary}
There is no $o({n\over \sum_{i=1}^n a_i})$ time randomized
approximation scheme algorithm for the bin packing problem.
\end{corollary}

\begin{proof}
It follows from Theorem~\ref{strong-lower-bound-theorem}.
\end{proof}


\section{ Generalization of  the Deterministic
Algorithm}\label{deterministic-section}

In this section, we generalize the existing deterministic
algorithm~\cite{FernadezLueker81} to handle the bin packing problem
with multiple sizes of bins. The bin packing problem is under a more
general version that allows different size of bins with different
weights (costs). The results of this section are used as submodules
in both sublinear time algorithms and streaming algorithms.

\begin{definition}\label{basic-def}\scrod
\begin{itemize}
\item
For an item $y$ and an integer $h$, define $y^h$ to be $h$ copies of
item $y$.
\item
A type $T_i$ of a bin of size $s$ is represented by
$(a_1^{b_{1,i}},\ldots , a_t^{b_{t,i}})$, which satisfies
$\sum_{j=1}^t b_{j,i}a_i\le s$. A bin of type $T_i$ can pack
$b_{1,i}$ items of size $a_1,\ldots ,$, and $b_{t,i}$ items of size
$a_t$. We use $w_{T_i}$ to represent the weight of a bin of type
$T_i$.
\end{itemize}
\end{definition}

It is easy to see that an optimal bin packing with
$(c,\eta,k)$-related bins only uses bins with
$s_{i_1}<s_{i_2}<\ldots <s_{i_k}$ with $w_{i_1}<w_{i_2}<\ldots
<w_{i_k}$. The classical bin packing problem only has one kind of
bins of size $1$. It is the bin packing problem with the
$(1,1,1)$-related bins. In the rest of this paper, a bin packing
problem without indicating $(c,\eta,k)$-related bins means the
classical bin packing problem.

\begin{lemma}\label{lp-lemma}Assume that $c$, $\eta$, and $k$ are constants. Assume that $\delta$ is a constant.
Given a bin packing  problem with $(c,\eta,k)$-related bins for
$B=\{ a_1^{n_1},\ldots , a_m^{n_m}\}$ with each $a_i\ge \delta$,
there is a $m^{O({1\over \delta})}$ time algorithm to give a
solution $(x_1,\ldots , x_q)$ with at most
$Opt_{c,\eta,k}(B)+\sum_{i=1}^qw_{T_i}$, where $x_i$ is the number
of bins of type $T_i$, and  $q$ is the number of types to pack items
of sizes in $\{a_1,\ldots , a_m\}$ with $q\le km^{{1\over \delta}}$.
\end{lemma}

\begin{proof}
 Since
$a_i$ is at least $\delta$, the number of items in each bin is at
most ${1\over \delta}$. Therefore, the number of types of bins is at
most $km^{1\over \delta}$. Let $T_1,\ldots , T_q$ be the all of the
possible types of bins to pack the items of size $a_1,\ldots , a_m$.

Let $x_i$ be the number of bins with type $T_i$. We define the
linear programming conditions:
\begin{eqnarray}
\min \sum_{i=1}^q w_{T_i}x_i\ \ &&\ \mbox{ subject to } \label{lp-1}\\
&&\sum_{i=1}^q b_{j,i} x_i\ge n_j\ \mbox{ for } j=1,2,\ldots , m\\
&&x_i\ge 0. \label{lp-3}
\end{eqnarray}

After obtaining the optimal solution $(x_1^*,\ldots , x_q^*)$ of the
linear programming, the algorithm outputs $(x_1,\ldots ,
x_q)=(\ceiling{x_1^*}, \ldots , \ceiling{x_q^*})$. Since
$\ceiling{x_i^*}\le x_i^*+1$, the cost for $(x_1,\ldots , x_q)$ is
at most $Opt_{c,\eta,k}(B)+\sum_{i=1}^qw_{T_i}$.

\vskip 10pt

{\bf Algorithm Pack-Large-Items($c, \eta, k, B$)}

Input: parameters $c, \eta, k$ and a list $B=\{ a_1^{n_1}\ldots ,
a_{m}^{n_m}\}$ to be packed in $(c, \eta, k)$ related bins.

Output: an approximation for $Opt_{c,\eta,k}(B)$.

Steps:

\qquad Solve the linear programming~(\ref{lp-1})-(\ref{lp-3}) for
$x_1^*,\ldots , x_q^*$.

\qquad Let $x_i=\ceiling{x_i^*}$ for $i=1,\ldots , q$.

\qquad Output $(x_1,\ldots , x_q)$.

{\bf End of Algorithm}

\end{proof}

With a constant $\epsilon$ to control the approximation ratio, we
define the following constants for
Lemma~\ref{convert-approximation-lemma}. We will also define a
threshold $\delta$ to control the size of large items. Let

\begin{eqnarray}
\mu&:=&{\epsilon \delta \eta\over 15},\label{mu-def}\\
 \epsilon_1&:=&{\epsilon\over \epsilon+2},\ \ {\rm and}\label{epsilon1-def}\\
 m&:=&{18\over \delta\eta}\ceiling{\epsilon_1^{-2}}.\label{m-def}
\end{eqnarray}

\begin{lemma}\label{large-packing-lemma}
 Assume that $c$, $\eta$, and $k$ are positive
constants, and $\epsilon$ and $\delta$ are constants in $(0,1)$.
Assume that the input list is $S$ for bin packing problem with
$(c,\eta,k)$-related bins and the size of each item in $S$ is at
least $\delta$. Let $\epsilon$ be a constant in $(0,1)$. The
constants $\delta, \mu, \epsilon_1$, and $m$ are given according to
equations (\ref{mu-def}) to (\ref{m-def}).  Let $h=\floor{n\over
m}$. Then there exists an $O(n)$ time algorithm that gives an
approximation $app$ with $Opt_{c,\eta,k}(S)\le app\le
(1+\epsilon)Opt_{c,\eta,k}(S)$ for all large $n$, where $n=|S|$.
\end{lemma}

\begin{proof}
Assume that $a_1\le a_2\le \ldots \le a_{n}$ is the increasing order
of all input elements at least $\delta$ with $n'=|S_{\ge \delta}|$.
 Let $L_0=a_1\le a_2\le \ldots \le a_n$. We partition
$a_1\le a_2\le \ldots \le a_n$ into $A_1y_1A_2y_2\ldots  A_m y_mR$
such that each $A_i$ has exactly $h-1$ elements and $R$ has less
than $h$ elements.

Using algorithm the classical algorithm, we can find the $ih$-th
element $y_i$ each in $O(n)$ time.

Consider the bin packing  problems: $L_1= y_1^{h} y_2^{h}\ldots
y_m^{h}$. We show that there is a small difference between the
results of two bin packing  problems for $L_0$ and $L_1$.

1) Assume that $L_0$ has a bin packing  solution. It can be
converted into a solution for $L_1$ via an adaption to that of $L_0$
(see Definition~\ref{basic-def}) with a small number of additional
bins.

Use the lots for the elements between $y_i$ and $y_{i+1}$ in $L_0$
to store the elements of $y_i$s, there are at most $2h$ $y_i$s left.
Therefore, we only have at most $2h$ elements left. The number of
bins for packing those left items is at most $2h$, which cost at
most $2h$ since $1$ is the maximal cost of one bin.

2) Assume that $L_1$ has a bin packing  solution. It can be
converted into a solution for $L_0$ with  a small number of
additional bins.

We use the lots for $y_i$ to store the elements between $y_{i-1}$
and $y_i$. We have at most $2h$ elements left, which cost at most
$2h$ since $1$ is the maximal cost of one bin.

The optimal number bins $Opt_{c,\eta,k}(L_0)$ for packing $L_0$ is
at least $mh\delta$, which have cost at least $mh\delta\eta$.
Therefore, we have
\begin{eqnarray}
Opt_{c,\eta,k}(L_0)\ge mh\delta\eta \label{L0-lower-bound-ineqn0}
\end{eqnarray}
Let $App(L_0)$ be an approximation for $L_0$ and $App(L_1)$ be an
approximation for $L_1$. We can obtain an
$(1+\epsilon/2)$-approximation $App(L_1)$ for packing $L_1$ by
Lemma~\ref{lp-lemma}. We have that

\begin{eqnarray*}
App(L_0)&=& App(L_1)+2h\ \ \label{L1-L0-app-eqn}\\
&\le&(1+\epsilon/2)Opt_{c,\eta,k}(L_1)+2h\\
&\le&((1+{\epsilon/2})(Opt_{c,\eta,k}(L_0)+2h)+2h\\
&\le&((1+{\epsilon/2})Opt_{c,\eta,k}(L_0)+6h\\
&\le&(1+\epsilon/2)Opt_{c,\eta,k}(L_0)+Opt_{c,\eta,k}(L_0)({6h\over mh\delta\eta})\ \ ({\rm by\ inequality}~(\ref{L0-lower-bound-ineqn0}))\\
&=&(1+\epsilon/2)Opt_{c,\eta,k}(L_0)+Opt_{c,\eta,k}(L_0){6\over m\delta\eta}\\
&\le&(1+\epsilon/2)Opt_{c,\eta,k}(L_0)+Opt_{c,\eta,k}(L_0)(\epsilon/2)\ \ \ \ \ \mbox{(by\ equations~(\ref{mu-def})\ to\ (\ref{m-def}).)}\\
&\le&(1+{\epsilon})Opt_{c,\eta,k}(L_0).
\end{eqnarray*}

By the analysis at case 2), if $App(L_1)\ge Opt_{c,\eta,k}(L_1)$,
we also have that the cost $App(L_1)+2h$ is enough to pack all items
in $L_0$. Therefore,
\begin{eqnarray}
App(L_0)\ge Opt_{c,\eta,k}(L_0). \label{L0-lowerbound-ineqn}
\end{eqnarray}

For a bin $b_i$, let $l(b_i)$ be the sum of sizes of items packed in
it.

\vskip 10pt

 {\bf Algorithm Packing($L_0$)}

Input: a list $L_0:=\{ a_1\ldots  a_{m}\}$

Output: an approximation for $Opt_{c,\eta,k}(L_0)$.

Steps:

\qquad Find the $ih$-th element $y_i$ in $L_0$ for $i=1,\ldots ,m$.

\qquad Let $L_1:= y_1^{h} y_2^{h}\ldots  y_m^{h}$.

\qquad Let $(x_1,\ldots , x_q):=$Pack-Large-Items($1, 1, 1, L_1$)
(see Lemma~\ref{lp-lemma}).

\qquad Let $App(L_1):=\sum_{i=1}^q w_{T_i} x_i$.

\qquad Convert $App(L_1)$ to $App(L_0)$ according to
equation~(\ref{L1-L0-app-eqn}).

\qquad\label{bin-list}
 Let $B=b_1,\ldots , b_u$ be the list of bins used for packing
(each $b_i$ has $l(b_i)$ available).

\qquad Output $App(L_0),$ and list $B$ of bins.

{\bf End of Algorithm}

We note that the list of bins $b_1,\ldots , b_u$ with their used
space $l(b_i)$ for each bin can be computed in $O(n)$ time from the
conversion based on $(x_1,\ldots , x_q)$ for $q$ types $T_1,\ldots ,
T_q$.

\end{proof}

\begin{lemma}[\cite{FernadezLueker81}]\label{linear-time-packing}
Let $\beta$ be a constant in $(0,1)$. Then there exists an $O(n)$
time algorithm that gives an approximation $app$ for packing $S$
with $Opt(S)\le app\le (1+\beta)Opt(S)+1$  for all large $n$.
\end{lemma}

\begin{proof}  The bin
packing problem is the same as the regular bin packing problem that
all bins are of the same size $1$. The problem is to minimize the
total number bins to pack all items.  We consider the approximation
to pack the small items after packing large items.

Assume that the input list is $S$ for bin packing problem. Let $S_{<
\delta}$ be the items of size less than $\delta$, and $S_{\ge
\delta}$ be the items of size at least $\delta$. Let $\delta$ be a
constant with $\delta\le {\beta\over 4}$.


\vskip 10pt

{\bf Algorithm Linear-Time-Packing$(n,S)$}

Input: A list of items $S=a_1\ldots  a_n$ and its number of items
$n$.

Output: an approximation for $Opt(S)$.

Steps:
\begin{enumerate}[1.]
\item
\qquad Let $App(S_{\ge \delta})$ and the bin list $b_1,\ldots , b_u$
be the output from calling Packing$(S_{\ge \delta})$ (see
Lemma~\ref{large-packing-lemma}).

\item\label{ass-s1'}
\qquad for $i=1$ to $u$
\item
\qquad \qquad If $l(b_i)\le 1-\delta$

\item
 \qquad \qquad Fill items from $S_{<\delta}$ into $b_i$ until less than $\delta$
 space left in $b_i$ or all items in $S_{<\delta}$ are packed.

\item
\qquad If there are some items of size less than $\delta$ left
\item
\qquad Then pack them into some bins so that at most one bin having
more than $\delta$ space used.

\item
\qquad Output the total number of bins used.

\end{enumerate}
{\bf End of Algorithm}

\vskip 10pt

Assume that an optimal solution of a bin packing  problem has two
types of bins. Each of the first type contains at least one item of
size at least $\delta$, and each of the second type only contain
items of size less than $\delta$. Let $V_1$ be the set of first type
bins, and $V_2$ be the set of all second type bins. Let
$U=App(S_{\ge \delta})$ be an $(1+\beta)$-approximation for packing
the first type of items. We have $|U|\le (1+\beta)|V_1|$.


Fill all items into those bins in $U$ so that each bin has less than
$\delta$ left. Put all of the items less than $\delta$ into some
extra bins, and at most one of them has more than $\delta$ space
left.

Case 1. If $U$ can contain all items, we have that $|U|\le
(1+\beta)|V_1|\le (1+\beta)|V_1\cup V_2|=(1+\beta)Opt(S)$.

Case 2. There is a bin beyond those in $U$ is used. Let $U'$ be all
bins without more than $\delta$ space left. We have that $|U'|\le
{|V_1\cup V_2|\over (1-\delta)}\le (1+\beta)|V_1\cup
V_2|=(1+\beta)Opt(S)$. Therefore, the approximate solution is at
most $(1+\beta)|V_1\cup V_2|+1=(1+\beta)Opt(S)+1$.

\end{proof}

\section{ Randomized Offline Algorithm}\label{rand-alg-sec}

In this section, we present sublinear time approximation schemes in
the offline model.

\subsection{Selecting Items from A List}\label{select-crucial-sec}

In this section, we show how a randomized algorithm to select some
crucial items from a list. Those elements are used for converting
the packing large items into linear programming as described in
Section~\ref{deterministic-section}.

In order to let linear programming have a small number of cases, the
$ih$-th elements are selected for $i=1,2,\ldots , m$, where the
large items are grouped into $m$ groups with $h$ items each. The
approximate $ih$-th elements (for $i=1,\ldots , m$) have similar
performance as the exact $ih$-th elements in the linear programming
method. The approximate $ih$-th elements (for $i=1,\ldots , m$) can
be obtained via sampling small number of items. The $ih$-th element
among the large items is approximated by the $ih$-th element among
the random samples from large items in the input list. The detail of
the algorithm is given at Select-Crucial-Items(.).

For a finite  set $A$, let $|A|$ be the number of elements in $A$.
For a list $L$ of items $a_1,\ldots, a_n$, let $|L|=n$.

\begin{definition}\label{rank-def}
Assume that $L=a_1,\ldots , a_n$ is the list of real numbers, and
$x$ is an integer.
\begin{itemize}
\item
Define  $\Rank(x, L)$ in $a_1,\ldots , a_n$ to be the interval
$[a,b]$  such that $|\{i: a_i<x\}|=a-1$ and $|\{i: a_i\le x \}|=b$.
Define $\minRank(x,L)$ to be $a$ and $\maxRank(x,L)$ to be $b$.

\item
Define  $\Rank_{\delta}(x, L)$ in $a_1,\ldots , a_n$ to be the
interval $[a,b]$  such that $|\{i: a_i<x\ and \ a_i\ge\delta
\}|=a-1$ and $|\{i: a_i\le x \ and \ a_i\ge\delta\}|=b$. Define
$\minRank_{\delta}(x,L)$ to be $a$ and $\maxRank_{\delta}(x,L)$ to
be $b$.

\item
$L[s,t]=a_s,a_{s+1},\ldots , a_{t}$ for $0<s\le t\le n$.
\end{itemize}
\end{definition}

\begin{definition}\label{s-<delta-def}
Assume that $S$ is a list of items for a bin packing problem and
$\delta$ is a real number. Define {\it $S_{< \delta}$} to  be the
sublist of the items of size less than $\delta$ in $S$, and {\it
$S_{\ge \delta}$} to be the sublist of the items of size at least
$\delta$ in $S$.
\end{definition}

By the definitions~\ref{rank-def} and \ref{s-<delta-def}, we have
\begin{eqnarray}
\minRank_{\delta}(x,L)&=&\minRank(x,L_{\ge \delta}),\\
\maxRank_{\delta}(x,L)&=&\maxRank(x,L_{\ge \delta}),\ \ \
\mbox{and}\\
\Rank_{\delta}(x,L)&=&\Rank(x,L_{\ge \delta}).
\end{eqnarray}

Let $m$ be a parameter at most $n$ and let
\begin{eqnarray}
h=\floor{n\over m}. \label{h-def-eqn}
\end{eqnarray}

 Let the sorted input list is partitioned into
$K_1 K_2\ldots  K_mR$ such that $|K_1|=|K_2|=\ldots  =|K_m|=h$, and
$0\le |R|<h$.

 \vskip 10pt

{\bf Algorithm Select-Crucial-Items($m, \alpha, \mu, X$)}

Input: two constants $\alpha$ and $\mu$ in $(0,1)$, an integer
parameter $m$ at least $2$, and a list $X=x_1, x_2,\ldots , $ is  a
finite list of random elements in $A$.

Steps:

\begin{enumerate}[1.]
\item\label{gamma-assign-in-Select-Crucial}
\qquad Select $\gamma={\mu \over 4m}$.
\item\label{set-c0-in-Select-Crucial-Items}
\qquad Select constant $c_0$ and $u=\ceiling{{c_0\log m\over
\gamma^2}}$ such that $2me^{-{\gamma^2 u\over 3}}<\alpha$ and $3\le
\gamma u$.


\item
\qquad If $v<u$ or $|X|<u$, then output $\emptyset$ and stop the
algorithm.

\item\label{set-pi-in-Select-Crucial-Items}
\qquad Let $p_i:={i\over m}$ for $i=1,\ldots , m$.
\item
\qquad Let $y_i$ ($i=1,\ldots , m$) be the least element $x_j$ such
that $|\{t: x_t \ \mbox{is\ in\ }\ X[1,u] \ \mbox{and}\ x_t\le
x_j\}|\ge \ceiling{{p_i}u}$.

\item
\qquad Output $(y_1,\ldots , y_m)$.
\end{enumerate}

{\bf End of Algorithm} \vskip 10pt

Lemma~\ref{select3-lemma} shows the performance of the algorithm
Algorithm Select-Crucial-Items(.). It is a step to convert the step
for packing large items into a dynamic programming method. When the
input list of items is $S$, the list $A$ in
Lemma~\ref{select3-lemma} is the sublist $S_{\ge \delta}$ of all
items of $S$ with size at least $\delta$, which will be specified in
the full algorithm. The random items $X$ is generated from the
subset of all random items of sizes at least $\delta$ in a set of
random items in $S$.

\begin{lemma}\label{select3-lemma} Let $\mu$ and $\alpha$ be positive
constants in $(0,1)$. Assume that $A$ is an input list of $n$
numbers of size at least $\delta$ with $n\ge {3(m+1)^2\over\mu}$.
Then the algorithm Select-Crucial-Items$(.)$ runs in   $O({m^2(\log
m)^2)\over \mu^2})$ time such that given a list $X$ of at least
${c_1m^2\log m\over ^2}$ random elements from $A$, it generates
elements $y_1\le \ldots \le y_m$ from the input list such that
$\prob[\Rank(y_i,A)\cap [ih-\mu h,ih+\mu h]]=\emptyset$ for at least
one $i\in \{1,\ldots , m\}]\le \alpha$, where $c_1=16c_0$, and
$c_0$ is the constant defined in Select-Crucial-Items(.), and $m$ is
an integer at most $n$.
\end{lemma}

\begin{proof}The algorithm  probabilistic
performance is analyzed with Chernoff bounds. Note that the number
of items $n$ in $A$ is not an input of this algorithm. We only use
it in the analysis, but not in the algorithm. Without loss of
generality, we assume $|X|=u$, where $u$ is defined in statement
\ref{set-c0-in-Select-Crucial-Items} in the Algorithm
Select-Crucial-Items(.).

 According to the algorithm $u=\ceiling{{c_0\over \gamma^2}\log
m}=\ceiling{16c_0m^2\log m\over \mu^2}=\ceiling{c_1m^2\log m\over
\mu^2}$.  We assume the number of random items in $X$ is at least
$u$.  By the equation~(\ref{h-def-eqn}) and the fact $m\le n$, we
have
\begin{eqnarray}
h&\le& {n\over m}\le h+1\le 2h, \ \ \ \mbox{and}
\label{range-n-m-h-ineqn}\\
 {h\over n}&\le& {1\over
m}.\label{h-m-n-ineqn}
\end{eqnarray}
By statement~\ref{gamma-assign-in-Select-Crucial} in
Select-Crucial-Items(.) and inequality (\ref{range-n-m-h-ineqn}), we
have ${n\over m}\le 2h$ and
\begin{eqnarray}
2\gamma \le {\mu \over 2m} \le  {\mu h \over  n}.\label{mu3-eqn}
\end{eqnarray}

Assume $\maxRank(y_i,A)<ih-\mu h$. We have that
\begin{eqnarray}
{\maxRank(y_i,A)\over n}&<&{ih-\mu h\over n}\\
&=&{ih\over n}-{\mu h\over n}\\
&\le& {i\over m}-{\mu h\over n}\ \ \ \mbox{(by inequality~(\ref{h-m-n-ineqn}))}\\
&\le& p_i-{\mu h\over n}.\label{maxRank-ineqn}
\end{eqnarray}

Let $p_i':=p_i-{\mu h \over n}> {\maxRank(y_i,A)\over n}$ (by
inequality (\ref{maxRank-ineqn})). By
Corollary~\ref{chernoff-lemma-a}, with probability at most
$e^{-{\gamma^2 u\over 3}}$, we have $ |\{j: x_j\in X[1,u]\ \mbox{
and }\ x_j\le y_i\}|$ to be at least
\begin{eqnarray*}
({\maxRank(y_i,A)\over n}+\gamma) u&<&p_i' u+\gamma u\\
&=&(p_i-({\mu h \over n}))u+\gamma u\\
&=& p_i u-({\mu h \over
n}-\gamma )u\\
&\le& p_i u-\gamma u \ \ (\mbox{by \ inequality}~(\ref{mu3-eqn}))\\
&\le& \ceiling{p_i u}
\end{eqnarray*}

Assume $\minRank(y_i,A)>ih+\mu h$. We have that
\begin{eqnarray}
{\minRank(y_i,A)\over n}&>&{ih+\mu h\over n}\\
&=&{ih\over n}+{\mu h\over n}\ \ \ \label{minRank-ineqn1}\\
&\ge& {i\over m}-{i\over n}+{\mu h\over n}\ \ \ \mbox{(by\ equation\ (\ref{h-def-eqn}))}\label{minRank-ineqn2}\\
&\ge& p_i-{i\over n}+{\mu h\over n}. \label{minRank-ineqn3}
\end{eqnarray}

Note that the transition from inequality~(\ref{minRank-ineqn1}) to
inequality~(\ref{minRank-ineqn2}) is due to equation
(\ref{h-def-eqn})), which implies $h\ge {n\over m}-1$ and ${h\over
n}\ge {1\over m}-{1\over n}$.

 Let $p_i'':=p_i-{i\over n}+{\mu h\over n}< {\minRank(y_i,A)\over n}$ (by inequality (\ref{minRank-ineqn3}). Note that $p_i$ is defined at
line~\ref{set-pi-in-Select-Crucial-Items} in Algorithm
Select-Crucial-Items(.). By Lemma~\ref{chernoff-lemma-a}, with
probability at most $P_{1,i}=e^{-{\gamma^2 u\over 3}}$, we have $
|\{j: x_j\in X[1,u]\ \mbox{ and }\ x_j\le y_i\}|$ to be at most

\begin{eqnarray}
({\maxRank(y_i,A)\over n}-\gamma) u&\ge&p_i'' u-\gamma u\\
 &=&(p_i-{i\over n}+{\mu h\over n})u-\gamma u\\
&\ge& p_i u+({\mu h \over n}-{i\over n}-\gamma )u\\
&\ge& p_i u+({\mu h \over n}-{m\over n}-\gamma )u\\
&\ge& p_i u+({\mu h \over 3n}-{m\over n})u+({2\mu h \over 3n}-\gamma )u\label{remove-m-ineqn}\\
&\ge& p_i u+0+({2\mu h \over 3n}-\gamma )u\label{remove-m-ineqn2}\\
&\ge& p_i u+({4\gamma\over 3}-\gamma )u\label{remove-m-ineqn3}\\
&\ge& p_i u+{\gamma\over 3}u\label{remove-m-ineqn4}\\
&\ge& p_i u+1\label{remove-m-ineqn5}\\
&>&\ceiling{p_i u}.
\end{eqnarray}

Note that $i\le m$.  The transition from
inequality~(\ref{remove-m-ineqn}) to
inequality~(\ref{remove-m-ineqn2}) is due to the condition $n\ge
{3(m+1)^2\over\mu}$, which implies that $h\ge  {n\over m}-1\ge
{3(m+1)\over\mu}-1\ge {3m\over\mu}+{3\over \mu}-1>{3m\over\mu}$. The
transition from inequality~(\ref{remove-m-ineqn2}) to
inequality~(\ref{remove-m-ineqn3}) is because of
inequality~(\ref{mu3-eqn}). The transition from
inequality~(\ref{remove-m-ineqn4}) to inequality
(\ref{remove-m-ineqn5}) is due to the setting in
statement~\ref{set-c0-in-Select-Crucial-Items} in
Select-Crucial-Items(.).

Therefore, with probability at most
$\sum_{i=1}^m(P_{i,1}+P_{i,2})\le 2me^{-{\gamma^2 u\over
3}}<\alpha$, $\Rank(y_i,A)\cap [ih-\mu h,ih+\mu h]=\emptyset$ for at
least one $i\in \{1,\ldots , m\}$.

\end{proof}

 \subsection{Packing Large Items and Small Items}

 In this section, we show how to pack large items from sampling
 items in the input list. Then we show how to pack small items after packing large items.

\begin{lemma}\label{convert-approximation-lemma} Assume that $c$, $\eta$, and $k$ are positive
constants, and $\epsilon$ and $\delta$ are constants in $(0,1)$ and
$\theta$ is a constant in $[0,1)$. Assume that the input list is $S$
for a bin packing problem with $(c,\eta,k)$-related bins.  The
constants $\delta, \mu, \epsilon_1$, and $m$ are given according to
equations (\ref{mu-def}) to (\ref{m-def}).  Assume that $n_{\ge
\delta}'$ is an approximation of $|S_{\ge \delta}|$ satisfying
\begin{eqnarray}
&&(1-\theta)|S_{\ge \delta}|\le n_{\ge \delta}'\le (1+\theta)|S_{\ge \delta}|,\label{tight-approx-theta-ineqn}\\
&&{36\theta \over \delta\eta}\le \epsilon, \ \mbox{and}\label{small-theta-ineqn}\\
&& {\theta \floor{|S_{\ge \delta}|\over m}}\ge 1\ \  \mbox{if}\
\theta>0. \ \ \ \label{large-n-n'-ineqn}
\end{eqnarray}
Let $h=\floor{|S_{\ge \delta}|\over m}$,  $h'=\floor{n_{\ge
\delta}'\over m}$, and $S'$ be a list of items of size less than
$\delta$. Assume that we have the following inputs available:
\begin{itemize}
\item
Let $y_1',\ldots , y_m'$ be a list of items from $S_{\ge \delta}$
such that $\Rank(y_i',S_{\ge \delta})\cap [ih-\mu h, ih+\mu
h]\not=\emptyset$ for $i=1,2,\ldots , m$
\item
An approximate solution for bin packing  with items in $B=\{y_1'^{h'
},\ldots , y_m'^{h'}\}\cup S'$ in $(c,\eta,k)$-related bins with
cost at most $(1+\epsilon)Opt_{c,\eta,k}(B)$
\end{itemize}
Then
there Packing-Conversion(.) is an $O(1)$ time algorithm  that gives
an approximation $app$ with $Opt_{c,\eta,k}(S_{\ge \delta}\cup
S')\le app\le (1+5\epsilon)Opt_{c,\eta,k}(S_{\ge \delta}\cup S')$.
\end{lemma}

\begin{proof}
Assume that $a_1'\le a_2'\le \ldots \le a_{n_{\ge \delta}}'$ is the
increasing order of all input elements of size at least $\delta$
with $n_{\ge \delta}=|S_{\ge \delta}|$.  Let
\begin{eqnarray}
L_{*}=a_1'\le a_2'\le \ldots \le a_{n_{\ge \delta}}'\cup
S'.\label{L*-def-eqn}
\end{eqnarray}
Let
\begin{eqnarray}
L_0=a_1'\le a_2'\le \ldots \le a_{n_{\ge \delta}'}'\cup
S'.\label{L0-def-eqn}
\end{eqnarray}
Note that in the case $n_{\ge \delta}'>|S_{\ge \delta}|$, we let
$a_{|S_{\ge \delta}|+1}'=\ldots =a_{n_{\ge \delta}'}'=1$ in list
$L_0$.
 Partition $a_1'\le
a_2'\le \ldots \le a_{n_{\ge \delta}}'$ into $A_1y_1A_2y_2\ldots A_m
y_mR$ such that each $A_i$ has exactly $h-1$ elements and $R$ has
less than $h$ elements.  Partition $a_1'\le a_2'\le \ldots \le
a_{n_{\ge \delta}'}'$ into $A_1y_1A_2y_2\ldots A_{m'} y_{m'}R'$ such
that each $A_i$ has exactly $h-1$ elements and $R'$ has less than
$h$ elements. We have
\begin{eqnarray}
m'&=&\floor{n_{\ge \delta}'\over h}\\
&\ge&\floor{(1-\theta)n_{\ge \delta}\over h}\\
&\ge&{(1-\theta)n_{\ge \delta}\over h}-1\\
&\ge&(1-\theta)\floor{n_{\ge \delta}\over h}-1\\
&\ge&(1-2\theta)\floor{n_{\ge \delta}\over h}\ \ \ \ \ \mbox{(by\ inequality~(\ref{large-n-n'-ineqn}))}\\
&=&(1-2\theta)m.\label{m'-lower-ineqn}
\end{eqnarray}

\begin{eqnarray}
m'&=&\floor{n_{\ge \delta}'\over h}\\
&\le&\floor{(1+\theta)n_{\ge \delta}\over h}\\
&\le&{(1+\theta)n_{\ge \delta}\over h}+1\\
&\le&(1+\theta)\floor{n_{\ge \delta}\over h}+1\\
&\le&(1+2\theta)\floor{n_{\ge \delta}\over h}\ \ \ \ \ \mbox{(by\ inequality~(\ref{large-n-n'-ineqn}))}\\
&=&(1+2\theta)m.
\end{eqnarray}


We have
\begin{eqnarray}
h'&=&\floor{n_{\ge \delta}'\over m}\le \floor{(1+\theta)n_{\ge
\delta}\over m}\\
&\le& \floor{(1+\theta)(h+1)}\\
&\le& \floor{(1+\theta)(h+\theta h)}\ \ \ \mbox{(by\ inequality\ (\ref{large-n-n'-ineqn}))}\\
&\le& \floor{(1+\theta)^2h}\\
&\le& (1+\theta)^2h\\
&\le& (1+3\theta)h.
\end{eqnarray}

We have
\begin{eqnarray}
h' &\ge& \floor{(1-\theta)n_{\ge \delta}\over m}\ge
\floor{(1-2\theta)n_{\ge \delta}+\theta n_{\ge \delta}\over m}\label{2-theta-tran-ineqn}\label{h-h'-start}\\
&\ge& \floor{(1-2\theta){n_{\ge \delta}\over m}+{\theta n_{\ge \delta}\over m}}\\
&\ge& \floor{(1-2\theta)\floor{n_{\ge \delta}\over m}+{\theta n_{\ge \delta}\over m}}\\
&\ge& \floor{(1-2\theta)\floor{n_{\ge \delta}\over m}+1}\\
&\ge& (1-2\theta)\floor{n_{\ge \delta}\over m}\\
&\ge& (1-2\theta)h.\label{2-theta-tran-ineqn2}
\end{eqnarray}

The transition from inequality~(\ref{2-theta-tran-ineqn}) to
inequality~(\ref{2-theta-tran-ineqn2}) is due to the fact ${\theta
n_{\ge \delta}\over m}\ge 1$ by
inequalities~(\ref{tight-approx-theta-ineqn}) and
(\ref{large-n-n'-ineqn}). By inequalities~(\ref{h-h'-start}) to
(\ref{2-theta-tran-ineqn2}), we have
\begin{eqnarray}
(1+3\theta)h\ge h' \ge (1-2\theta)h. \label{2-theta-tran-ineqn2b}
\end{eqnarray}
Inequality~(\ref{2-theta-tran-ineqn2b}) also holds if $\theta=0$.


Consider the bin packing  problems: $L_1= y_1'^{h'} y_2'^{h'}\ldots
y_m'^{h'}\cup S'$.
We show that there is a small difference between the results of two
bin packing  problems for $L_0$ and $L_1$.


{\bf Claim~\ref{convert-approximation-lemma}.1.} For every solution
of cost $x$ with $(c,\eta,k)$-related bins for list $L_0$, there is
a solution of cost at most $x+(10\theta +4\mu )mh+4h$ for list
$L_1$.

\begin{proof}
Assume that $L_0$ has a bin packing  solution. It can be converted
into a solution for $L_1$ via an adaption to that of $L_0$ with a
small number of additional bins.

We use the lots for the elements in $A_{i+1}y_{i+1}$ in $L_0$ to
store the elements of $y_i'$s. By
inequality~(\ref{2-theta-tran-ineqn2b}) and the assumption
$\Rank_{\delta}(y_i',S_{\ge\delta})\cap [ih-\mu h, ih+\mu
h]\not=\emptyset$ for $i=1,2,\ldots , m$, there are at most
$(3\theta+2\mu)h$ $y_i'$s left as unpacked for each $y_i'^{h'}$ with
$i\le m'$. Therefore, we only have that the number of  elements left
as unpacked in $L_1$ is at most
\begin{eqnarray*}
&&m'(3\theta+2\mu )h+(|m-m'|+2)h'\\
&\le& (3\theta+2\mu )(1+2\theta)mh+(2\theta m+2)(1+3\theta)h\ \ \ \
\ \mbox{(by\ inequality~(\ref{2-theta-tran-ineqn2b})\ and
(\ref{m'-lower-ineqn}))}.
\end{eqnarray*}
 The number of bins for packing those left items is at most
$(3\theta+2\mu )(1+2\theta)mh+(2\theta m+2)(1+3\theta)h$. Since $1$
is the maximal cost of one bin, the cost for packing the left items
at most
\begin{eqnarray*}
&&(3\theta+2\mu )(1+2\theta)mh+(2\theta m+2)(1+3\theta)h\\
&\le&
2(3\theta+2\mu )mh+2(2\theta m+2)h\ \ \ \ \ \ \mbox{(by\ inequality~(\ref{small-theta-ineqn}))}\\
&\le& (10\theta +4\mu )mh+4h.
\end{eqnarray*}
\end{proof}

{\bf Claim~\ref{convert-approximation-lemma}.2.} For every solution
of cost $y$ with $(c,\eta,k)$-related bins for list $L_1$, there is
a solution of cost at most $y+(\mu+2\theta)m h+2h$ for list $L_0$.

\begin{proof} Assume that $L_1$ has a bin packing  solution. It can be
converted into a solution for $L_0$ with  a small number of
additional bins.

We use the lots for $y_i'^{h'}$ to store the elements in $A_iy_i$.
We have at most $(\mu+2\theta) h$ elements left for each $A_iy_i$.
Totally, we have at most $m(\mu+2\theta) h+2h$ items left. The bins
for packing those left items is at most $m(\mu+2\theta) h+2h$, which
cost at most $m(\mu+2\theta) h+2h$ since $1$ is the maximal cost of
one bin.
\end{proof}

The optimal number bins $Opt_{c,\eta,k}(L_0)$ for packing $L_0$ is
at least $mh\delta$, which have cost at least $mh\delta\eta$.
Therefore, we have
\begin{eqnarray}
Opt_{c,\eta,k}(L_0)\ge mh\delta\eta \label{L0-lower-bound-ineqn}
\end{eqnarray}
For an approximation $App(L_1)$  for packing $L_1$, let
\begin{eqnarray}
App(L_0)=App(L_1)+(\mu+2\theta)m h+2h \label{app-L0-eqn}
\end{eqnarray}
 be an approximation for
packing $L_0$ by Claim~\ref{convert-approximation-lemma}.2. We have
that

\begin{eqnarray}
&&Opt_{c,\eta,k}(L_0)\\
&\le& App(L_0) \\
&=& App(L_1)+m(\mu+2\theta) h+2h  \ \  \ \ \ \mbox{(by\ equation\ (\ref{app-L0-eqn}))}\label{L0-L1-approximation-eqn}\\
&\le&(1+\epsilon)Opt_{c,\eta,k}(L_1)+m(\mu+2\theta) h+2h\\
&\le&((1+{\epsilon})(Opt_{c,\eta,k}(L_0)+(10\theta+4\mu m)m h+4h)+m(\mu+2\theta) h+2h\\
&& \ \ \ \mbox{(by\ Claim~\ref{convert-approximation-lemma}.1)}\\
&\le&(1+{\epsilon})(Opt_{c,\eta,k}(L_0)+Opt_{c,\eta,k}(L_0)({5\mu mh+12\theta m h+6h\over mh\delta\eta}))\ \ ({\rm by\ inequality}~(\ref{L0-lower-bound-ineqn}))\\
&=&(1+{\epsilon})Opt_{c,\eta,k}(L_0)+Opt_{c,\eta,k}(L_0)({{5\mu\over
\delta\eta} +{12\theta \over
\delta\eta}+{6\over m\delta\eta}})\\
&\le&(1+{\epsilon})Opt_{c,\eta,k}(L_0)+Opt_{c,\eta,k}(L_0)({{5\mu\over
\delta\eta} +{12\theta \over
\delta\eta}+{\epsilon\over 3}})\ \ \ \ \ \mbox{(by~equation (\ref{m-def}))}\\
&\le&(1+{\epsilon})(Opt_{c,\eta,k}(L_0)+Opt_{c,\eta,k}(L_0)({{5\mu\over
\delta\eta}
+{\epsilon\over 3}+{\epsilon\over 3}}))\ \ \ \ \ \mbox{(by~inequality (\ref{small-theta-ineqn}))}\\
&\le&(1+{\epsilon})(Opt_{c,\eta,k}(L_0)+Opt_{c,\eta,k}(L_0)({\epsilon\over 3}+{\epsilon\over 3}+{\epsilon\over 3}))\ \ \ \mbox{(by~equation~(\ref{mu-def}))}\\
&\le&(1+\epsilon)(1+\epsilon)Opt_{c,\eta,k}(L_0)\\
&\le&(1+3\epsilon)Opt_{c,\eta,k}(L_0).
\end{eqnarray}


The list $L_*$ has at most $\theta n_{\ge \delta}$ more items than
$L_0$. Therefore
\begin{eqnarray}
Opt_{c,\eta,k}(L_*)&=&Opt_{c,\eta,k}(L_0)+\theta n_{\ge \delta}\\
&\le& Opt_{c,\eta,k}(L_0)+{\theta\over 1-\theta} n_{\ge \delta}'\\
&\le& Opt_{c,\eta,k}(L_0)(1+{\theta n_{\ge \delta}'\over (1-\theta)Opt_{c,\eta,k}(L_0)})\\
&\le& Opt_{c,\eta,k}(L_0)(1+{\theta n_{\ge \delta}'\over (1-\theta)mh\delta\eta})\ \ \ \mbox{(by\ inequality~(\ref{L0-lower-bound-ineqn}))}\\
&\le& Opt_{c,\eta,k}(L_0)(1+{\theta \over (1-\theta)h\delta\eta}\cdot {n_{\ge \delta}'\over m})\\
&\le& Opt_{c,\eta,k}(L_0)(1+{\theta \over (1-\theta)h\delta\eta}\cdot 2h)\ \ \ \mbox{(by~inequality~(\ref{range-n-m-h-ineqn}))}\\
&\le& Opt_{c,\eta,k}(L_0)(1+{2\theta \over (1-\theta)\delta\eta})\\
&\le& Opt_{c,\eta,k}(L_0)(1+{4\theta \over \delta\eta})\\
&\le& Opt_{c,\eta,k}(L_0)(1+\epsilon)\ \ \mbox{(by\
inequality~(\ref{small-theta-ineqn})}).\label{L*-L0-inequality}
\end{eqnarray}

Let
\begin{eqnarray}
App(L_*)=(1+\epsilon)App(L_0).\label{app-L0-L*-eqn}
\end{eqnarray}

 Therefore, we have $App(L_*)\ge
Opt_{c,\eta,k}(L_*)$ by inequality~(\ref{L0-lowerbound-ineqn}) and
inequality~(\ref{L*-L0-inequality}). On the other hand, we have
$App(L_*)=(1+\epsilon)App(L_0)\le
(1+\epsilon)(1+3\epsilon)Opt_{c,\eta,k}(L_0)\le
(1+\epsilon)(1+3\epsilon)Opt_{c,\eta,k}(L_*)\le
(1+5\epsilon)Opt_{c,\eta,k}(L_*).$
\end{proof}

\vskip 10pt

 {\bf Algorithm Packing-Conversion($n_{\ge\delta}',App(L_1)$)}

Input: an integer $n_{\ge\delta}'$ is an approximation to
$|S_{\ge\delta}|$ with $(1-\theta)|S_{\ge \delta}|\le
n_{\ge\delta}'\le (1+\theta)|S_{\ge \delta}|$, and
 an approximate solution $App(L_1)$  for the bin packing
with items in $L_1=\{y_1'^{h' },\ldots , y_m'^{h'}\}\cup S'$ in
$(c,\eta,k)$-related bins with cost at most
$(1+\epsilon)Opt_{c,\eta,k}(L_1)$, where
 $L_1=\{y_1'^{h'},\ldots , y_m'^{h'}\}\cup S'$ is a list of
items such that  $\Rank(y_i',S_{\ge \delta})\cap [ih-\mu h, ih+\mu
h]\not=\emptyset$ for $i=1,2,\ldots , m$, and $S'$ is a list of
items of size less than $\delta$.

 Output: an approximation for $Opt_{c,\eta,k}(L_*)$, where $L_*$ is defined by equation~(\ref{L*-def-eqn}).

Steps:

\qquad Convert the approximation of $App(L_1)$ to $App(L_0)$ as
equation~(\ref{app-L0-eqn}) in the proof.

\qquad Convert the approximation of $App(L_0)$ to $App(L_*)$ as
equation~(\ref{app-L0-L*-eqn}).


\qquad Output $App(L_*)$ 

{\bf End of Algorithm}

\vskip 10pt

\begin{lemma}\label{global-convert-lemma} Let $\xi$ be a small
constant in $(0,1)$.  Assume that $S_{\ge \varphi}$ is a list of
items of size at least $\varphi$, $S_{<\varphi}$ is a list of items
of size less than $\varphi$, and $S_{<\varphi}'$ is another list of
items of size less than $\varphi$. If $\sum_{a_i\in
S_{<\varphi}}a_i+ \sum_{a_i\in S_{\ge \varphi}}a_i\le
(1+\xi)(\sum_{a_i\in S_{<\varphi}'}a_i+ \sum_{a_i\in S_{\ge
\varphi}}a_i)$ and $\sum_{a_i\in S_{<\varphi}'}a_i+ \sum_{a_i\in
S_{\ge \varphi}}a_i\le (1+\xi)(\sum_{a_i\in S_{<\varphi}}a_i+
\sum_{a_i\in S_{\ge \varphi}}a_i)$, then $Opt(S_{<\varphi}\cup
S_{\ge \varphi})\le {1+\xi\over 1-\varphi}\cdot Opt(S_{<
\varphi}'\cup S_{\ge \varphi})+1$ and $Opt(S_{<\varphi}'\cup S_{\ge
\varphi})\le {1+\xi\over 1-\varphi}\cdot Opt(S_{<\varphi}\cup S_{\ge
\varphi})+1$.
\end{lemma}

\begin{proof}
Let $L=S_{<\varphi}\cup S_{\ge \varphi}$  and $L'=S_{<\varphi}'\cup
S_{\ge \varphi}$. Without loss of generality, let $Opt(L)\le
Opt(L')$. We just need to prove that $Opt(L')\le {1+\xi\over
1-\varphi}\cdot Opt(L)$.

 For a bin packing $P$ for $L$, we convert it
into another bin packing for $L'$ by increasing small number of
bins. At most one bin in $P$ wastes more than $\varphi$ space by
replacing the items in $S_{<\varphi}$ with those in $S_{<\varphi}'$.
If no additional bin is used for packing $L'$, we have $Opt(L')\le
Opt(L)$.

If some new bins are needed, the total number of bins is at most
\begin{eqnarray*}
{(\sum_{a_i\in S_{<\varphi}'}a_i+ \sum_{a_i\in S_{\ge
\varphi}}a_i)\over 1-\varphi}+1&\le& {(1+\xi)(\sum_{a_i\in
S_{<\varphi}}a_i+ \sum_{a_i\in S_{\ge \varphi}}a_i)\over
1-\varphi}+1\\
&\le& {1+\xi\over 1-\varphi}\cdot Opt(L)+1.
\end{eqnarray*}
Therefore, we have $Opt(L')\le {1+\xi\over 1-\varphi}\cdot Opt(L)$.
\end{proof}

The following Lemma~\ref{convert-approximation2-lemma} is only for
the classical bin packing problem that all bins are of the same size
$1$.

\vskip 10pt

{\bf Algorithm Packing-Small-Items$(X, s_1, S'')$}

Input: $X=(x_1,\ldots , x_q)$ for the $q$ types $T=\langle
T_1,\ldots , T_q\rangle$ for the $(1+\beta)$-approximation for
packing a list $S''=\{y_1'^{h' },\ldots , y_m'^{h' }\}$, and
$s_1=\sum_{a_i\in S'}a_i$ is the sum of sizes in list $S'$ of items
of size less than $\delta$.

Output: an approximation for $Opt(S''\cup S')$.

Steps:

\begin{enumerate}[1.]
\item\label{ass-s1'}
\qquad Let $s_1':=s_1$.
\item
\qquad Repeat\label{repeat-start}
\item
\qquad\qquad Let $i:=1$.
\item
\qquad \qquad For each type $T_i=(b_{1,i}a_1,\ldots , b_{m,i}a_m)$
(which satisfies $\sum_{j=1}^m b_{j,i}a_i\le 1$)

\item
 \qquad\qquad \qquad Let $t_i:=\sum_{j=1}^m b_{j,i}a_m$ and $h_i:=\max(1-\delta-t_i,0)$

\item
 \qquad\qquad \qquad ($h_i$ is the available space in a bin of type $T_i$ for packing
items of size $<\delta$).
\item
\qquad\qquad  Let $s_1':=\max(s_1'-x_ih_i,0)$ (fill each bin of type
$T_i$ with size $h_i$ of (fractional) items).
\item
\qquad\qquad Let $i:=i+1$.
\item
\qquad Until $s_1'=0$ or $i>q$.\label{repeat-end}

\item
\qquad If $s_1'>0$
\item
\qquad Then find the least number $k$ such that $k(1-\delta)\ge
s_1'$
\item
\qquad\qquad  and fill the (fractional) items left in $s_1'$ into
those $k$ bins.
\end{enumerate}
{\bf End of Algorithm}

\vskip 10pt

\begin{lemma}\label{convert-approximation2-lemma}
Let $\beta$ be a constant in $(0,1)$  with $\beta\le {1\over 2}$,
$\theta$ be a constant in $[0,1)$  with $\theta\le \beta$, and
$\delta$ be a constant with $\delta\le {\beta\over 4}$. Let $m$ and
$h'$ be
 integers.
  Let $S'$ be a list of items of size less than $\delta$.
Assume that $S''=\{y_1'^{h' },\ldots , y_m'^{h' }\}$ with $y_i'\ge
\delta$ for $i=1,\ldots ,m$ and $S''$ is large enough to satisfy
\begin{eqnarray}
h'm\ge{2\over \beta\delta}.\label{large-n-ineqn}
\end{eqnarray}
Then Packing-Small-Items(.) is an $O(q)$ time algorithm that given a
solution $(x_1,\ldots , x_q)$ for bin packing with items in $S''$
with the total number of bins at most $(1+\beta)Opt(S'')$, and
$s_1=\sum_{a_i\in S'}a_i$, where $x_i$ is the number of bins of type
$T_i$, and $q$ is the number of types to pack $y_1',\ldots , y_m'$
with $q\le m^{O({1\over \delta})}$ (see Lemma~\ref{lp-lemma}),
it gives an approximation $app$ for packing $S''\cup S'$ with
$Opt(S''\cup S')\le app\le (1+2\beta)Opt(S''\cup S')$.
\end{lemma}

\begin{proof}  The bin
packing problem is the same as the regular bin packing problem that
all bins are of the same size $1$. The problem is to minimize the
total number bins to pack all items. We consider the approximation
to pack the small items after packing large items.

Assume that an optimal solution of a bin packing  problem has two
types of bins. Each first type bin contains at least one item of
size $\delta$, and each second type bin only contains items of size
less than $\delta$. Let $V_1$ be the set of first type bins, and
$V_2$ be the set of all second type bins. Let $U$ be an
$(1+\beta)$-approximation for the items in $S''$. We have $|U|\le
(1+\beta)|V_1|$. Let $s_{large}=\sum_{a_i\in S''}a_i=\sum_{i=1}^q
h_i'y_i'$ and $s_{small}=\sum_{a_i\in S'}a_i=s_1$.


Fill items of size less than $\delta$ into those bins in $U$ so that
each bin has less than $\delta$ left. Put all of the items less than
$\delta$ into some extra bins, and at most one of them has more than
$\delta$ space left. We use a fractional way to pack small items.
Since each bin with small items has at least $\delta$ space left,
and each small item is of size at most $\delta$, the fractional
packing of small items can be converted into a non-fractional
packing. A similar argument is also shown in
Lemma~\ref{global-convert-lemma}.

Case 1. If $U$ can contain all items, we have that $|U|\le
(1+\beta)|V_1|\le (1+\beta)|V_1\cup V_2|$.

Case 2. There is a bin beyond those in $U$ is used. Let $U'$ be all
bins without more than $\delta$ space left. We have
\begin{eqnarray}
|U'|&\le& {s_{large}+s_{small}\over 1-\delta}\\
&\le& (1+{\delta\over 1-\delta})(s_{large}+s_{small})\\
&\le& (1+2\delta)(s_{large}+s_{small})\\
&\le& (1+\beta/2)(s_{large}+s_{small}).\label{U'-bound}
\end{eqnarray}

On the other hand, $|V_1\cup V_2|\ge s_{large}+s_{small}$.
Therefore, the approximate solution $|U'|+1$ has
\begin{eqnarray}
|U'|+1 &\le& (1+\beta/2)|V_1\cup
V_2|+1\ \ \ \mbox{(by inequality~(\ref{U'-bound}))}\\
&=&(1+\beta/2)Opt(S''\cup S')+1\label{change-to-6-epsilon-eqnuality0}\\
&\le &(1+1.5\beta)Opt(S''\cup
S') \ \ \ \mbox{(by inequality~(\ref{large-n-ineqn}))}\label{change-to-6-epsilon-eqnuality}\\
&\le &(1+2\beta)Opt(S''\cup S').
\end{eqnarray}
Packing the items in $S''$ needs at least $\delta mh'$ bins.
Therefore, the transition from
inequality~(\ref{change-to-6-epsilon-eqnuality0}) to
inequality~(\ref{change-to-6-epsilon-eqnuality}) is by the condition
in inequality~(\ref{large-n-ineqn}).
\end{proof}

\vskip 10pt

{\bf Algorithm Packing-With-Many-Large-Items$(\alpha,\beta, n, s_1,
n_{\ge \delta}',S)$}

Input: a parameter $\beta\in (0,1)$, $n_{\ge \delta}'$ is an
approximation to $|S_{\ge \delta}|$, and  $s_1$ is an approximation
for $\sum_{a_i\in S_{<\delta}}a_i$ with $(1-\xi)(\sum_{a_i\in
S}a_i)\le s_1+\sum_{a_i\in S_{\ge\delta}}a_i\le (1+\xi)(\sum_{a_i\in
S}a_i)$, $S$ is the list of input items $a_1,\ldots , a_n$ for bin
packing, and $n$ is the number of items in $S$.

Output: an approximation for $Opt(S)$.

Steps:

\begin{enumerate}[1.]
\item
\qquad Select an integer constant $d_1$ such that $g_1({1\over
2})^{d_1\over 1-\delta}<\alpha$.

\item\label{random-samples-Packing-With-Many-Large-Items}
\qquad Select a list $L_1$ of $2c_1d_1{n\over n_{\ge
\varphi}'}m^2\log m$ random elements in the input list $S$,

\qquad  where constant $c_1$ is defined in
Lemma~\ref{select3-lemma}, and constant $d_1$ is defined in
line~\ref{d1-def}.


\item
\qquad  Let $L_2$ be the list of items of size at least $\delta$ in
$L_1$.

\item
\qquad Let $(y_1',\ldots , y_m')$:=Select-Crucial-Items($m,\alpha,
\mu, L_2$) (see Lemma~\ref{select3-lemma}).


\item\label{define-h'-Approximate-Bin-Packing}
\qquad Let $X=(x_1,\ldots , x_q):=$Pack-Large-Items($1, 1, 1, B$)
with $B=\{y_1'^{h'},\ldots , y_m'^{h' }\}$

\qquad (see Lemma~\ref{lp-lemma}), where $h'=\floor{{n_{\ge
\varphi}'\over m}}$.

\item\label{line-app}
\qquad Let $App_1$:=Packing-Small-Items$(X, s_1, B)$ (see
Lemma~\ref{convert-approximation2-lemma}).

\item
\qquad Let $App_2:=$Packing-Conversion($n_{\ge \delta}', App_1$)(see
Lemma~\ref{convert-approximation-lemma}) for packing all items in
$S$.
\item
\qquad Output ${1+\xi\over 1-\delta}\cdot App_2$.
\end{enumerate}
{\bf End of Algorithm}

\vskip 10pt

\begin{lemma}\label{convert-approximation2b-lemma}
Assume that  $S$ is a list of items for bin packing problem.  Let
$\beta$ be a constant in $(0,1)$  with $\beta\le {1\over 2}$,
$\theta$ be a constant in $[0,1)$  with $\theta\le \beta$, $\delta$
be a constant with $\delta\le {\beta\over 4}$, $\xi$ be a constant
with $\xi\le {\beta\over 4}$, and constant $\epsilon=6\beta$. The
constants $\mu, \epsilon_1$, and $m$ are given according to
equations (\ref{mu-def}) to (\ref{m-def}). Assume that $n_{\ge
\delta}'$ is an approximation of $|S_{\ge \delta}|$ satisfying the
inequalities~(\ref{tight-approx-theta-ineqn}),
(\ref{small-theta-ineqn}), (\ref{large-n-n'-ineqn}), and
(\ref{large-n-ineqn}).  Assume that $s_1$ is an approximation for
$\sum_{a_i\in S_{<\delta}}a_i$ with $(1-\xi)(\sum_{a_i\in S}a_i)\le
s_1+\sum_{a_i\in S_{\ge\delta}}a_i\le (1+\xi)(\sum_{a_i\in S}a_i)$.
Then Packing-With-Many-Large-Items(.) is an $O({n\over \sum_{i=1}^n
a_i}+O({1\over \beta})^{O({1\over\beta})})$ time algorithm that
gives an approximation $app$ for packing $S$ with $Opt(S)\le app\le
(1+16\beta)Opt(S)$  with the failure probability at most $\alpha$.
\end{lemma}

\begin{proof} The bin
packing problem is the same as the regular bin packing problem that
all bins are of the same size $1$. The problem is to minimize the
total number bins to pack all items. We consider the approximation
to pack the small items after packing large items.

We sample some random items of size at least $\varphi$ from the
input list $S$. When  an item from the input list $S$ is randomly
selected, an item of size at least $\varphi$ has an equal
probability, which is defined by the $p_{\varphi}$ below:
\begin{eqnarray}
p_{\varphi}={|\{i: a_i\ge \varphi\  and \ a_i\in \{a_1,\ldots ,
a_n\}\}|\over n}={n_{\ge\varphi}\over n}.\label{p-varphi-def-eqn2}
\end{eqnarray}

By inequality~(\ref{tight-approx-theta-ineqn}) and
equation~(\ref{p-varphi-def-eqn2}), we have
\begin{eqnarray}
p_{\varphi}{n\over n_{\ge \varphi}'}\ge {1\over
1+\theta}.\label{p-varphi-n-n-varphi-bound-ineqn}
\end{eqnarray}
By Theorem~\ref{chernoff3-theorem}, with probability at most
$g_1({1\over 2})^{p_{\varphi}{2c_1d_1n\over n_{\ge \varphi}'}m^2\log
m}\le g_1({1\over 2})^{{2c_1d_1\over 1+\delta}m^2\log m}<\alpha$
(see line~\ref{d1-def} in Approximate-Bin-Packing$(.)$), we cannot
obtain at least
\begin{eqnarray}
(1-{1\over 2})p_{\varphi}({2c_1d_1n\over n_{\ge \varphi}'}m^2\log
m)&\ge& p_{\varphi}{n\over n_{\ge \varphi}'}(c_1d_1m^2\log
m)\\
&\ge& {1\over 1+\theta}\cdot c_1d_1m^2\log m\ \ \ \mbox{(by\ inequality\ (\ref{p-varphi-n-n-varphi-bound-ineqn}))}\\
&\ge& c_1m^2\log m
\end{eqnarray}
 random elements of size at least $\varphi$ by
sampling $2c_1d_1{n\over n_{\ge\varphi}'}m^2\log m$ elements.

By Lemma~\ref{select3-lemma}, with probability at most $\alpha$, we
cannot obtain the list $y_1\le \ldots \le y_m$ from the input list
such that $\Rank(y_i,S_{\ge\varphi})\cap [ih-\mu h,ih+\mu
h]\not=\emptyset$ for all $i\in \{1,\ldots , m\}$ in $O({m^2(\log
m)^2)\over \mu^2})$ time using
 ${c_1m^2\log m\over \mu^2}$ random
elements from the input.

Therefore, we have probability at most $\alpha+\alpha+\alpha\le
{1\over 4}$, the following (a) or (b) is false:

(a). Statements~\ref{item1-app-sum-lemma},
~\ref{item2-app-sum-lemma}, and~\ref{item3-app-sum-lemma} of
Lemma~\ref{app-sum-lemma} are true.

(b). $\Rank(y_i,S_{\ge \varphi}) \cap [ih-\mu h,ih+\mu
h]\not=\emptyset$ for all $i\in \{1,\ldots , m\}$.

Assume that both statements (a) and (b) are true in the rest of the
proof. This makes the analysis of algorithm become deterministic.

Imagine that $S_1'$ is a list of items of size less than $\delta$
and has $s_1=\sum_{a_i\in S_1'}a_i$.
 By
Lemma\ref{convert-approximation2-lemma}, line~\ref{line-app} gives
$App_1$ to be an $(1+2\beta)$-approximation for packing $S''\cup
S_1'$.

By  Lemma~\ref{convert-approximation-lemma}, $App_2$ is an
$(1+5\times 2\beta)$-approximation for packing $S_{\ge\delta}\cup
S_1'$.

By Lemma~\ref{global-convert-lemma}, ${1+\xi\over 1-\delta}\cdot
App_2$ is an ${1+\xi\over 1-\delta}\cdot(1+10\beta)$-approximation
for packing $S_{\ge\delta}\cup S_{<\delta}=S$. We note that
\begin{eqnarray*}
{1+\xi\over 1-\delta}\cdot(1+10\beta)&\le& (1+{\xi+\delta\over
1-\delta})\cdot(1+10\beta)\\
&\le& (1+2(\xi+\delta))\cdot(1+10\beta)\\
&\le& (1+\beta)\cdot(1+10\beta)\\
&\le& (1+\beta++10\beta+10\beta^2)\\
&\le& (1+\beta++10\beta+5\beta)\ \ \ \mbox{(note that $\beta\le {1\over 2})$}\\
&\le& 1+16\beta.
\end{eqnarray*}
Thus, ${1+\xi\over 1-\delta}\cdot App_2$ is an
$(1+16\beta)$-approximation for packing $S_{\ge\delta}\cup
S_{<\delta}=S$.

The function  is executed under the condition that $n_{\ge
\varphi}'=\Omega(\sum_{i=1}^n a_i)$.
Statement~\ref{random-samples-Packing-With-Many-Large-Items} takes
$O({n\over n_{\ge\delta}'})=O({n\over \sum_{i=1}^na_i})$ time. The
computational time at
statement~\ref{define-h'-Approximate-Bin-Packing} is $({1\over
\beta})^{O({1\over\beta})}$ which follows from Lemma~\ref{lp-lemma}.
The other statements only takes $O(1)$ time.
\end{proof}

The following Lemma~\ref{convert-approximation3-lemma} is only for
the classical bin packing problem that all bins are of the same size
$1$.

 \vskip 10pt

{\bf Algorithm Packing-With-Few-Large-Items$(\xi, x, s_1)$}

Input: a small parameter $\xi\in [0,1)$, an integer $x$ with $x\le
\xi \sum_{i=1}^n a_i$ and $x\ge |S_{\ge\delta}|$, and a real
 $s_1$ with
$(1-\xi)(\sum_{a_i\in S}a_i)\le s_1+\sum_{a_i\in
S_{\ge\delta}}a_i\le (1+\xi)(\sum_{a_i\in S}a_i)$. ($s_1$ is an
approximate sum of sizes of small items of size at most $\delta$).

Output: an approximation for $Opt(S)$.

Steps:

\begin{enumerate}[1.]
\item\label{find-k-line}
\qquad Find the least number $k$ such that $k(1-\delta)\ge s_1$

\qquad\qquad (the $k$ bins are for packing items of size less than
$\delta$).
\item
\qquad Output ${1+\xi\over 1-\delta}\cdot(k+x+1)$ for packing $S$
($x$ bins are for packing items of size $\ge \delta$)
\end{enumerate}
{\bf End of Algorithm}

 \vskip 10pt

\begin{lemma}\label{convert-approximation3-lemma}
Assume that  $S$ is a list of items for bin packing problem. Let
$\delta$ be a constant in $(0,1)$. Assume that we have the following
inputs available:
\begin{itemize}
\item
$x$ is an approximation for $|S_{\ge\delta}|$ with $x\le \xi
\sum_{i=1}^n a_i$ and $x\ge |S_{\ge\delta}|$ for some small $\xi\in
(0,1)$.
\item
 $s_1$ is an approximation for $\sum_{a_i\in S_{<\delta}}a_i$ with
$(1-\xi)(\sum_{a_i\in S}a_i)\le s_1+\sum_{a_i\in
S_{\ge\delta}}a_i\le (1+\xi)(\sum_{a_i\in S}a_i)$.
\end{itemize}
and the parameters satisfy the following conditions
\begin{eqnarray}
 \delta&\le& {1\over 4}, \label{alpha-delta-ineqn3}\\
  \xi&\le& {1\over 4},  \ \ \ \mbox{and}\label{alpha-delta-ineqn3b}\\
2&<&\delta \sum_{i=1}^n a_i. \label{alpha-delta-ineqn4}
\end{eqnarray}
Then Packing-With-Few-Large-Items(.) is an $O(1)$ time algorithm
that gives an approximation $app$ for packing $S$ with $Opt(S)\le
app\le (1+8(\delta+\xi))Opt(S)$.
\end{lemma}

\begin{proof}   The bin
packing problem is the same as the regular bin packing problem that
all bins are of the same size $1$. The problem is to minimize the
total number bins to pack all items.

Imagine $S_{<\delta}'$ is a list of elements of size less than
$\delta$ and $\sum_{a_i\in S'}a_i=s_1$. Let $S'=S_{<\delta}'\cup
S_{\ge \delta}$. Let $s_0=\sum_{i=1}^n a_i$ to be the sum of sizes
of input items. By line~\ref{find-k-line} in
Packing-With-Few-Large-Items(.), we have
\begin{eqnarray}
k+x&\le&{s_1\over 1-\delta}+1+x\\
&\le& {s_0(1+\xi)\over 1-\delta}+\xi s_0+1\\
&\le& ({1+\xi\over 1-\delta}+\xi) s_0+1.
\end{eqnarray}


Furthermore, assume that the inequalities~(\ref{alpha-delta-ineqn3})
to (\ref{alpha-delta-ineqn4}) holds. We have
\begin{eqnarray*}
({1+\xi\over 1-\delta}+\xi)
&\le&(1+{\xi+\delta\over 1-\delta}+\xi)\\
&\le&(1+2(\xi+\delta)+\xi)\\
&=&(1+2\delta+3\xi)
\end{eqnarray*}

Therefore,  we have
\begin{eqnarray*}
k+x+1&\le&  (1+2\delta+3\xi)Opt(S)+1\\
&\le& (1+3\delta+3\xi)Opt(S)  \ \ \ \ \ \mbox{(by\ inequality\
(\ref{alpha-delta-ineqn4}))}.
\end{eqnarray*}
By Lemma~\ref{global-convert-lemma}, we have the ${1+\xi\over
1-\delta}(1+3\delta+3\xi)$-approximation for packing $S$. We note
that
\begin{eqnarray*}
{1+\xi\over 1-\delta}(1+3\delta+3\xi)&\le& (1+{\xi+\delta\over
1-\delta})(1+3\delta+3\xi)\\
&\le& (1+2(\xi+\delta))(1+3\delta+3\xi)\\
&\le& 1+2(\xi+\delta)+(3\delta+3\xi)+2(\xi+\delta)(3\delta+3\xi)\\
&\le& 1+2(\xi+\delta)+3(\delta+\xi)+3(\delta+\xi)\\
&\le& 1+8(\xi+\delta).
\end{eqnarray*}

\end{proof}

\subsection{Full Sublinear Time Approximation Scheme for Bin Packing
}\label{full-sublinear-section}

Now we present a sublinear time approximation scheme for the bin
packing problem. The brief idea of our sublinear time algorithm is
given in Section~\ref{idea-overview-sec}. After setting up some
parameters, it divides the interval $(0,1]$ for item sizes into
$O(\log n)$ intervals $(0,1]=I_1\cup \ldots  \cup I_k$, called a
$(\varphi, \delta, \gamma)$-partition as described in
section~\ref{adaptive-sampling-sec}. Applying the algorithm
described in section~\ref{adaptive-sampling-sec}, we get an
approximation about the distribution of the items in the intervals
$I_1,\ldots , I_k$. If the total size $\sum_{i=1}^n a_i$ is too
small, for example $O(1)$, the linear time algorithm described in
section~\ref{deterministic-section} is used to output an
approximation for the bin packing problem. Otherwise, we give a
sublinear time approximation for the bin packing problem. In order
to pack large items, we derive the approximate crucial items, which
are the approximate $ih$-th elements among the large items of size
at least $\varphi$ for $i=1,\ldots , m$, where $h$ and $m$ are
defined in equations (\ref{h-def-eqn}), and ((\ref{m-def})),
respectively. The algorithm described in
section~\ref{deterministic-section} is used to pack large items. The
small items are filled into bins which have space left after packing
large items, and some additional fresh bins. With the approximate
sum of sizes of small items, we can calculate the approximate number
of fresh bins to be needed to pack them. If the total sum of the
sizes of large items is too small to affect the total approximation
ratio, we just directly pack the small items according to
approximate sum of the sizes of those small objects.

\vskip 10pt

{\bf Algorithm Approximate-Bin-Packing$(\tau, n, S)$}

Input: a positive real number $\tau $, an integer $n$, and  a list
$S$ of $n$ items $a_1,\ldots , a_n$ in $(0,1]$.

Output: an approximation $app$ with $Opt(S)\le app\le (1+\tau
)Opt(S)+1$.

Steps:

\begin{enumerate}[1.]

\item\label{beta-assingment-line}
\qquad Let $\beta:={\tau \over 30}$ and $\epsilon:=6\beta$.

\item\label{delta-assingment-line}
\qquad Let $\delta:={\epsilon\over 4}$ and $\theta:={\epsilon
\delta\over 36}$.

\item
\qquad Let $\mu,\epsilon_1$ and $m$ are selected by
equations~(\ref{mu-def}), (\ref{epsilon1-def}), and (\ref{m-def}),
respectively.

\item\label{c-eta-k-setting-line}
\qquad Let $c:=\eta:=k:=1$ (classical bin packing).

\item
\qquad Let $\alpha:=1/12$.

\item\label{varphi-assignment-line}
\qquad Let $\varphi:=\delta$.

\item
\qquad Let $\gamma:=\delta^3$.

\item\label{d1-def}
\qquad Select an integer constant $d_1$ such that $g_1({1\over
2})^{d_1\over 1-\delta}<\alpha$.

\item\label{partition-line}
\qquad Derive a $(\varphi, \delta, \gamma)$-partition $P=I_1\cup
\ldots \cup I_k$ for $(0,1]$.

\item\label{Approximate-Interval-line}
\qquad Let $(s, s_1, n_{\ge
\varphi}')$:=Approximate-Interval$(\varphi, \delta, \gamma, \theta,
\alpha, P, n,S)$ (see Lemma~\ref{app-sum-lemma}).


\item\label{if-s-small-inequality}
\qquad If $s<\max(({4m\over \theta \delta^2}), ({4\over
\delta^2}\cdot {(1+\theta)m\over \theta}), ({16\over
\delta^2}\cdot{(1+\theta)\over \beta\delta}))$
\item
\qquad then
\item\label{line-time-line}
\qquad\qquad Output Linear-Time-Packing$(n,S)$ (see
Lemma~\ref{linear-time-packing})  and terminate the algorithm.

\item\label{if-condition}
\qquad If $n_{\ge \varphi}'\ge {\delta^2\over 4}s$
\item
\qquad then

\item\label{if-end}
\qquad\qquad Output Packing-With-Many-Large-Items$(\alpha,\beta, n,
s_1, n_{\ge \delta}',S)$ (see
Lemma~\ref{convert-approximation2b-lemma}).

\item
\qquad else

\item
\qquad\qquad If $n_{\ge \varphi}'>0$

\item
\qquad\qquad then let $x:={n_{\ge \varphi}'\over 1-\theta}$ and
$\xi:=\max(\delta^2,\theta+\delta^3)$

\item
\qquad\qquad else let $x:=6\delta s$ and
$\xi:=\max(12\delta,\theta+\delta^3)$.

\item\label{execute-Packing-With-Few-Large-Items}
\qquad\qquad Output Packing-With-Few-Large-Items$(\xi, x, s_1)$ (see
Lemma~\ref{convert-approximation3-lemma}).

\end{enumerate}


{\bf End of Algorithm}
\vskip 10pt

\begin{proof}[Theorem~\ref{random-approx-theorem}] 
Calling function Approximate-Interval$(.)$ in
line~\ref{Approximate-Interval-line} in the algorithm
Approximate-Bin-Packing$(.)$,   we obtain  $s$ for an approximate
sum $\sum_{i=1}^na_i$ of items in list $S$,  $s_1$ for an
approximate sum
 of items in list $S_{<\varphi}$, an approximate
number $n_{\ge \varphi}'$ of items of size at least $\varphi$ (see
Lemma~\ref{app-sum-lemma}). With probability at most $\alpha$, at
least one of statements~\ref{item1-app-sum-lemma},
~\ref{item2-app-sum-lemma}, \ref{item3-app-sum-lemma},
\ref{item3b-app-sum-lemma}, and \ref{item4-app-sum-lemma} of
Lemma~\ref{app-sum-lemma} of Lemma~\ref{app-sum-lemma} is false.
Therefore, we have probability at most $\alpha$, the following
statement (a) is false:

(a). Statements~\ref{item1-app-sum-lemma},
~\ref{item2-app-sum-lemma}, \ref{item3-app-sum-lemma},
\ref{item3b-app-sum-lemma}, and \ref{item4-app-sum-lemma} of
Lemma~\ref{app-sum-lemma} are true.

Assume that statement (a) is true in the rest of the proof. By
statement~\ref{item1-app-sum-lemma} of Lemma~\ref{app-sum-lemma}, we
have that if $n_{\ge\varphi}'>0$, then
\begin{eqnarray}
(1-\theta)n_{\ge\varphi}\le n_{\ge\varphi}'\le
(1+\theta)n_{\ge\varphi}.\label{n-varphi-accruacy-ineqn}
\end{eqnarray}
Let $s_0=\sum_{i=1}^n a_i$. By line~\ref{Approximate-Interval-line}
in Approximate-Bin-Packing(.) and Lemma~\ref{app-sum-lemma}, $s$ is
an approximation of $s_0=\sum_{i=1}^n a_i$, $s_1$ is an
approximation of $\sum_{i=1, a_i<\varphi}^n a_i$, and $n_{\ge
\varphi}'$ is an approximation of the number $n_{\ge \varphi}$ of
items of size at least $\varphi$.  A
$(\varphi,\delta,\gamma)$-partition for $(0,1]$ divides the interval
$(0,1]$ into intervals $I_1=[\pi_1, \pi_0], I_2=(\pi_2, \pi_1],
I_3=(\pi_3,\pi_2],\ldots ,I_k=(0, \pi_{k-1}]$ as in
Definition~\ref{partition-def}.

{\bf Claim~\ref{random-approx-theorem}.1.} If the condition in
line~\ref{if-s-small-inequality} of Approximate-Bin-Packing(.) is
true, the algorithm outputs an approximation $app(S)$ for the bin
packing problem $S$ with $Opt(S)\le app(S)\le (1+\tau)Opt(S)+1$.

\begin{proof}
 We note that if the condition in
line~\ref{if-s-small-inequality} is true, then $s=O(1)$ since
$\beta, \theta, m,$ and $\delta$ are all constants. By
statement~\ref{item3-app-sum-lemma} of Lemma~\ref{app-sum-lemma}, we
have $s_0=O(1)$. In this case, we use the linear time deterministic
algorithm by Lemma~\ref{linear-time-packing}, which warrants the
desired ratio of approximation.
\end{proof}

In the rest of the proof, we assume that the condition in
line~\ref{if-s-small-inequality} is false. We have the inequality:
\begin{eqnarray}
s\ge\max(({4m\over \theta \delta^2}), ({4\over \delta^2}\cdot
{(1+\theta)m\over \theta}), ({16\over \delta^2}\cdot{(1+\theta)\over
\beta\delta})).\label{s-first-lower-bound-ineqn}
\end{eqnarray}

By inequality~(\ref{s-first-lower-bound-ineqn}), we have the
inequality
\begin{eqnarray}
s\ge {8\over \delta^2}\cdot {1\over \beta\delta}\ge {8\over
\delta^3}.\label{s-lowerbound-ineqn}
\end{eqnarray}
 Therefore,
\begin{eqnarray}
\delta\le {\delta ^4s\over 8}.\label{delta-s-ineqn}
\end{eqnarray}

By statement~\ref{item3-app-sum-lemma} of Lemma~\ref{app-sum-lemma},
we have $s\le (1+\theta)(\sum_{i=1}^na_i)=(1+\theta) s_0$. By
inequality~(\ref{s-lowerbound-ineqn}) and the fact $\delta\le 1$ (by
the setting in line~\ref{delta-assingment-line}), we have
\begin{eqnarray}
s_0\ge {s\over 1+\theta}\ge {s\over 2}\ge {4\over \delta^3}\ge
4\label{s0-lower-bound-inequality}.
\end{eqnarray}

By inequality~(\ref{s0-lower-bound-inequality}) and
statement~\ref{item3b-app-sum-lemma} of Lemma~\ref{app-sum-lemma},
we have
\begin{eqnarray}
{1\over 4}(1-\theta)(1-\delta)\varphi(\sum_{i=1}^na_i)\le s\le
(1+\theta)s_0.\label{s0-s-lower-b-ineqn}
\end{eqnarray}

{\bf Claim~\ref{random-approx-theorem}.2.} If the condition at
line~\ref{if-condition} of the algorithm Approximate-Bin-Packing(.)
is true, then  with failure probability at most $\alpha$, the
algorithm outputs an approximation $app(S)$ for the bin packing
problem with $Opt(S)\le app(S)\le (1+\tau)Opt(S)+1$.

\begin{proof}
Assume that the condition at line~\ref{if-condition} of the
algorithm Approximate-Bin-Packing(.) is true.  This is the case that
the number of large items is large. The condition of
line~\ref{if-s-small-inequality} in Approximate-Bin-Packing(.) is
false. Since condition of line~\ref{if-condition} in
Approximate-Bin-Packing(.) is true,  we have
\begin{eqnarray}
h'm&\ge& \floor{{n_{\ge \varphi}'\over m}}m\\
&\ge& ({n_{\ge \varphi}'\over m}-1)m\\
&=& n_{\ge \varphi}'-m\label{h'm-ineqn0}\\
&\ge& {\delta^2\over 4}s-m\label{h'm-ineqn1}\\
&\ge& {\delta^2\over 4}s-{s\over 64}\ \ \ \ \ \mbox{(by inequality\
(\ref{s-first-lower-bound-ineqn}))}\label{h'm-ineqn1b}\\
&\ge& {\delta^2\over 8}s\label{h'm-ineqn1c}\\
&\ge& {2\over \beta\delta}, \ \ \ \ \ \mbox{(by inequality\
(\ref{s-first-lower-bound-ineqn}))}\label{h'm-ineqn2}
\end{eqnarray}
 where $h'$ is defined is
statement~\ref{define-h'-Approximate-Bin-Packing} of
Packing-With-Many-Large-Items(.).  Note that  the transition from
inequality (\ref{h'm-ineqn0}) to inequality (\ref{h'm-ineqn1}) is
due to condition of line~\ref{if-condition} in
Approximate-Bin-Packing(.) is true, and the transition from
inequality (\ref{h'm-ineqn1}) to inequality (\ref{h'm-ineqn2}) is
due to inequality~(\ref{s-first-lower-bound-ineqn}),
 Thus, the
inequality~(\ref{large-n-ineqn}) condition in
Lemma~\ref{convert-approximation-lemma} is true.

Inequality (\ref{tight-approx-theta-ineqn}) is satisfied because of
inequality (\ref{n-varphi-accruacy-ineqn}). Inequality
(\ref{small-theta-ineqn}) is satisfied because of the setting in
lines~\ref{beta-assingment-line} to~\ref{c-eta-k-setting-line} of
Approximate-Bin-Packing(.).
 We have the inequality
\begin{eqnarray}
\theta\floor{n_{\ge\varphi}\over m}&\ge&
\theta\floor{n_{\ge\varphi}'\over (1+\theta)m}\label{theta-times-n-varphi-ineqn1}\\
&\ge&
\theta\floor{\delta^2s\over 4(1+\theta)m}\label{theta-times-n-varphi-ineqn2}\\
&\ge&
\theta\floor{\delta^2\cdot ({4\over \delta^2}\cdot {(1+\theta)m\over \theta})\over 4(1+\theta)m} \label{theta-times-n-varphi-ineqn3}\\
&\ge&
\theta\floor{(1+\theta)\over \theta}\\
&\ge&
\theta\floor{{1\over \theta}+1}\\
&\ge& \theta\cdot {1\over \theta}=1.
\end{eqnarray}

The transition from inequality (\ref{theta-times-n-varphi-ineqn1})
to inequality (\ref{theta-times-n-varphi-ineqn2}) is because the
condition of statement~\ref{if-condition} of
Approximate-Bin-Packing(.)) is true. The transition from inequality
(\ref{theta-times-n-varphi-ineqn2}) to inequality
(\ref{theta-times-n-varphi-ineqn3}) is because of
inequality~(\ref{s-first-lower-bound-ineqn}). Thus,
inequality~(\ref{large-n-n'-ineqn}) is satisfied.

By Lemma~\ref{convert-approximation2b-lemma}, the algorithm gives an
approximation $app(S)$ such that $Opt(S)\le app(S)\le
(1+16\beta)Opt(S)\le (1+\tau)Opt(S)$ (by the setting of $\beta$ in
statement~\ref{beta-assingment-line} of Approximate-Bin-Packing(.))
with the failure probability at most $\alpha$.
\end{proof}

{\bf Claim~\ref{random-approx-theorem}.3.} If the condition at
line~\ref{if-condition} of the algorithm Approximate-Bin-Packing(.)
is false, then the algorithm outputs an approximation $app(S)$ for
the bin packing problem with $Opt(S)\le app(S)\le (1+\tau)Opt(S)+1$.

\begin{proof}
In the case that the condition at line~\ref{if-condition} does not
hold, we have that
\begin{eqnarray}
n_{\ge \varphi}'&<&{\delta^2\over 4}s\\
&\le& {\delta^2\over 4}(1+\delta)s_0\ \ \ \mbox{(by\ inequality~(\ref{s0-s-lower-b-ineqn}))}\\
&\le& {\delta^2\over 2}s_0.\label{n-ge-varphi-s0}
\end{eqnarray}


Line~\ref{execute-Packing-With-Few-Large-Items} in the algorithm
Approximate-Bin-Packing(.) will be executed.
By inequality~(\ref{s0-lower-bound-inequality}),
inequality~(\ref{alpha-delta-ineqn4}) is true.
Inequalities~(\ref{alpha-delta-ineqn3}) and
(\ref{alpha-delta-ineqn3b}) follow from
lines~\ref{beta-assingment-line} and \ref{delta-assingment-line} in
the Algorithm Approximate-Bin-Packing(.).

By statements~\ref{item1-app-sum-lemma}
and~\ref{item2-app-sum-lemma} of Lemma~\ref{app-sum-lemma}, we have
\begin{eqnarray}
s_1&=&\sum_{\hat{C}(I_j,S)>0\ and \ j>1}\hat{C}(I_j,S)\pi_j\\
&\ge&\sum_{\hat{C}(I_j,S)>0\ and \ j>1}(1-\theta)C(I_j,S)\pi_j\ \ \
\mbox{(by\ statement\ \ref{item1-app-sum-lemma}\ of\
Lemma~\ref{app-sum-lemma})}\\
&\ge&(1-\theta)\sum_{a_i\in I_j\ with\ \hat{C}(I_j,S)>0\ and \ j>1}a_i\\
&\ge&(1-\theta)\sum_{a_i\in I_j\ and \ j>1}a_i-\sum_{a_i\in I_j\ with\ \hat{C}(I_j,S)=0\ and \ j>1}a_i\\
&\ge&(1-\theta)\sum_{a_i\in S_{<\varphi}}a_i-({\delta^3\over
2}\sum_{a_i\in S_{< \varphi}}a_i+{\gamma\over n}).\ \ \ \mbox{(by\
statement\ \ref{item2-app-sum-lemma}\ of\
Lemma~\ref{app-sum-lemma})}\\
\end{eqnarray}

We have
\begin{eqnarray}
s_1+\sum_{a_i\in S_{\ge \varphi}} a_i&\ge& (1-\theta)(\sum_{a_i\in
S_{< \varphi}}a_i)-({\delta^3\over 2}\sum_{a_i\in S_{<
\varphi}}a_i+{\gamma\over
n})+\sum_{a_i\in S_{\ge \varphi}}a_i\\
&\ge& (1-\theta)(\sum_{a_i\in S}a_i)-({\delta^3\over 2}\sum_{a_i\in
S_{< \varphi}}a_i+{\gamma\over
n})\\
&\ge& (1-\theta)(\sum_{a_i\in S}a_i)-({\delta^3\over 2}\sum_{a_i\in
S}a_i+{\gamma\over n})\ \ \ \ \mbox{(note $S_{<\varphi}\subseteq S$)}\\
&\ge& (1-\theta-{\delta^3\over 2})(\sum_{a_i\in S}a_i)-{\gamma\over
n}\\
&\ge& (1-\theta-\delta^3)(\sum_{a_i\in S}a_i).\ \ \ \mbox{(by\
inequality~(\ref{s0-lower-bound-inequality}))}
\end{eqnarray}

By statements~\ref{item1-app-sum-lemma}
and~\ref{item2-app-sum-lemma} of Lemma~\ref{app-sum-lemma}, we have
\begin{eqnarray}
s_1&=&\sum_{\hat{C}(I_j,S)>0\ and \ j>1}\hat{C}(I_j,S)\pi_j\\
&\le&\sum_{\hat{C}(I_j,S)>0\ and \ j>1}(1+\theta)C(I_j,S)\pi_j\ \ \
\mbox{(by\ statement\ \ref{item1-app-sum-lemma}\ of\
Lemma~\ref{app-sum-lemma})}\\
&\le&{1+\theta\over 1-\varphi}\sum_{a_i\in I_j\ with\ \hat{C}(I_j,S)>0\ and \ j>1}a_i\\
&\le&{1+\theta\over 1-\varphi}\sum_{a_i\in S_{<\varphi}}a_i.
\end{eqnarray}

By statement~\ref{item1-app-sum-lemma} of Lemma~\ref{app-sum-lemma},
we have
\begin{eqnarray} s_1+\sum_{a_i\in S_{\ge \varphi}} a_i
&\le& {1+\theta\over 1-\varphi}\sum_{a_i\in S}a_i\\
&\le& (1+\theta)(1+2\varphi)\sum_{a_i\in S}a_i\\
&\le& (1+\theta+4\varphi)\sum_{a_i\in S}a_i.
\end{eqnarray}

Therefore,
\begin{eqnarray}
(1-(\theta+\delta^3))(\sum_{a_i\in S}a_i)&\le& s_1+\sum_{a_i\in
S_{\ge \varphi}} a_i \le (1+(\theta+4\varphi))(\sum_{a_i\in
S}a_i).\label{s1-left-right-bound-ineqn}
\end{eqnarray}

Since the condition at line~\ref{if-condition} in
Approximate-Bin-Packing(.) is false, we discuss two cases

\begin{itemize}
\item Case $n_{\ge \varphi}'>0$.

We have the inequalities
\begin{eqnarray}
\sum_{a_i\ge \varphi} a_i&\le&n_{\ge \varphi}\\
&\le&(1+\theta)n_{\ge\varphi}'\ \ \ \mbox{(by\ inequality~(\ref{n-varphi-accruacy-ineqn}))}\\
&\le&2n_{\ge\varphi}'\\
&\le& \delta^2 s_0.    \ \ \ \mbox{(by\
inequality~(\ref{n-ge-varphi-s0}) )}
\end{eqnarray}

By statement \ref{item1-app-sum-lemma} of Lemma~\ref{app-sum-lemma},
we have
\begin{eqnarray}
{n_{\ge \varphi}'\over 1-\theta}\ge
|S_{\ge\delta}|.\label{n'-lower-bound-for-n'>0}
\end{eqnarray}

 We also have
\begin{eqnarray}
{n_{\ge \varphi}'\over 1-\theta}&\le& {\delta^2\over 4(1-\theta)}s\label{n-s0-s-ineqn}\ \ \mbox{(line~\ref{if-condition}\ in\ Approximate-Bin-Packing(.)\ is\ false)}\\
&\le& {\delta^2\over 2}s\\
&\le& {\delta^2\over 2}(1+\delta)s_0\ \ \mbox{(by~inequality~(\ref{s0-s-lower-b-ineqn}))}\\
&\le& \delta^2 s_0.\label{n-s0-s-ineqn3}
\end{eqnarray}

In this case, $x= {n_{\ge \varphi}'\over 1-\theta}$ by inequality
(\ref{n'-lower-bound-for-n'>0}) and
inequalities~(\ref{n-s0-s-ineqn}) to (\ref{n-s0-s-ineqn3}),
 and
$\xi=\max(\delta^2,\theta+\delta^3)$ by
inequality~(\ref{s1-left-right-bound-ineqn}). They satisfy the
conditions of Lemma~\ref{convert-approximation3-lemma}, which
implies that
 the approximation ratio is  $(1+8(\delta+\xi))\le (1+\tau)$ by the
 assignments in lines \ref{beta-assingment-line} and \ref{delta-assingment-line} in algorithm
Approximate-Bin-Packing(.).

\item
Case $n_{\ge \varphi}'=0$.

 By statement \ref{item2-app-sum-lemma} of
Lemma~\ref{app-sum-lemma}, we have
\begin{eqnarray*}
\delta|S_{\ge \varphi}|&\le&\sum_{a_i\ge \varphi}a_i\\
&=&\sum_{a_i\in I_1}a_i\\
&\le& {\delta^3\over 2}s_0+{\gamma\over n}\ \ \ \mbox{(apply
statement~\ref{item2-app-sum-lemma}\ of\
Lemma~\ref{app-sum-lemma}\ with\ }\hat{C}(I_1,S)=n_{\ge \varphi}'=0)\\
&\le& {\delta^3\over 2}s_0+{\delta}\\
&\le& {\delta^3\over 2}s_0+{\delta^4 s\over 8}\ \ \mbox{(by~inequality~(\ref{delta-s-ineqn}))}\\
&\le& {\delta^3\over 2}s_0+{\delta^4 (1+\delta)s_0\over 8}\ \ \mbox{(by~inequality~(\ref{s0-s-lower-b-ineqn}))}\\
&\le& {\delta^3\over 2}s_0+{\delta^4 s_0\over 4}\\
&\le& {3\delta^3\over 4}s_0\\
&\le& {3\delta^3\over 4}{8s\over\delta}\ \ \ \mbox{(by\
inequality~(\ref{s0-s-lower-b-ineqn}))}\\
&\le& 6\delta^2 s.
\end{eqnarray*}
Therefore,
\begin{eqnarray}
|S_{\ge \varphi}|&\le& 6\delta s\label{s-S-s0-ineqn1}\\
&\le&6\delta(1+\delta)s_0\label{s-S-s0-ineqn2}\\
 &\le&12\delta s_0.\label{s-S-s0-ineqn3}
\end{eqnarray}
\end{itemize}

In this case, let $x= 6\delta s$ by
inequality~(\ref{s-S-s0-ineqn1}), and let
$\xi=\max(12\delta,\theta+\delta^3,1+\theta+4\varphi)$ by
inequality~(\ref{s1-left-right-bound-ineqn}) and
inequalities~(\ref{s-S-s0-ineqn1}) to (\ref{s-S-s0-ineqn3}). They
satisfy the conditions of Lemma~\ref{convert-approximation3-lemma},
which implies
 the approximation ratio is  $(1+8(\delta+\xi))\le (1+\tau)$ by the
 assignments in lines \ref{beta-assingment-line} and \ref{delta-assingment-line} in algorithm
Approximate-Bin-Packing(.).
This completes the proof of Claim~\ref{random-approx-theorem}.3.
\end{proof}

{\bf Claim~\ref{random-approx-theorem}.4.} The algorithm runs in
$O({n(\log n)(\log\log n)\over \sum_{i=1} a_i}+({1\over \tau
})^{O({1\over\tau })})$ time.

\begin{proof}
We give the computational time about the algorithm.
Lines~\ref{beta-assingment-line} to \ref{d1-def} take $O(1)$ time.
Line \ref{partition-line} takes $O(\log n)$ time.
 By
Lemma~\ref{app-sum-lemma}, Line~\ref{Approximate-Interval-line}
takes $O(({n\over \sum_{i=1}^n a_i})(\log n)\log\log n))$ time. Line
\ref{line-time-line} takes $O(n)$ time by calling
Linear-Time-Packing$(S)$ by Lemma~\ref{linear-time-packing}. This
only happens when $\sum_{i=1}^n a_i=O(1)$.

By Lemma~\ref{convert-approximation2b-lemma}, statement~\ref{if-end}
of Approximate-Bin-Packing(.) takes $O({n\over \sum_{i=1}^n
a_i}+O({1\over \beta})^{O({1\over\beta})})=O({n\over \sum_{i=1}^n
a_i}+O({1\over \tau})^{O({1\over\tau})})$ time.

Line~\ref{execute-Packing-With-Few-Large-Items} takes $O(1)$ time by
Lemma~\ref{convert-approximation3-lemma}. Therefore, in the worst
case, the algorithm takes $O({n(\log n)(\log\log n)\over \sum_{i=1}
a_i}+({1\over \tau })^{O({1\over\tau })})$ time.

{\bf Claim~\ref{random-approx-theorem}.5.} The failure probability
of the algorithm is at most ${1\over 4}$.

\begin{proof} Two statements~\ref{Approximate-Interval-line} and~\ref{if-end} in the algorithm  may fail due to randomization.
Each of them has probability at most $\alpha$ to fail by
Lemma~\ref{app-sum-lemma} (for statement (a)),  and
Claim~\ref{random-approx-theorem}.2. Therefore, the failure
probability of the entire algorithm is at most $2\alpha\le {1\over
4}$.
\end{proof}

\end{proof}

The theorem follows from the above claims. This completes the proof
of Theorem~\ref{random-approx-theorem}

\end{proof}



The following Theorem~\ref{dense-hierarchy-theorem} gives a dense
sublinear time hierarchy approximation scheme for bin packing
problem.

\begin{theorem}\label{dense-hierarchy-theorem}
For each $\epsilon\in (0,1)$, and $b\in (0,1]$, there is a randomize
$(1+\epsilon)$-approximation for all $\sum(n^{b})$-bin packing
problems in time $O(n^{1-b}(\log n)\log \log n)$ time, but there is
no $o(n^{1-b})$ time $(1+\epsilon)$-approximation algorithm
$\sum(n^{b})$-bin packing problem.
\end{theorem}

\begin{proof}
It follows from Theorem~\ref{random-approx-theorem} and
Theorem~\ref{strong-lower-bound-theorem}.
\end{proof}

\subsection{NP Hardness}

In this section, we show that $\sum(n^b)$  and $S(\delta)$ are both
NP-hard. We reduce the 3-partition problem, which is defined below,
to them.

\begin{definition}
The {\it 3-partition problem} is to decide whether a given multiset
of integers in the range $({B\over 4}, {B\over 2})$ can be
partitioned into triples that all have the same sum $B$, where $B$
is an integer. More precisely, given a multiset S  of $n = 3 t$
positive integers, can S be partitioned into m subsets $S_1, S_2,
\ldots , S_t$ such that the sum of the numbers in each subset is
equal?
\end{definition}

It is well known that 3-partition problem is
NP-complete~\cite{GareyJohnson}. It is used in proving the following
NP-hard problems (Theorem~\ref{NP-hard-sublinear-theorem} and
Theorem~\ref{delta-bin-packing-NP-theorem})

\begin{theorem}\label{NP-hard-sublinear-theorem}
For each constant $b\in (0,1)$, the bin packing problem in
$\sum(n^b)$ is NP-hard.
\end{theorem}

\begin{proof}
We construct a reduction from 3-partition problem to the
$\sum(n^b)$-bin packing problem via some padding. Assume that
$b_1,\ldots , b_n$ is a list of 3-partition problem with all items
in $({B\over 4}, {B\over 2})$. The bin packing problem for
$\sum(n^b)$ is constructed below:

It has a new list of elements: $a_1,\ldots , a_n, a_{n+1},\ldots ,
a_{m}$ such that $\sum_{i=1}^m a_i=m^b$, where $a_i={b_i\over B}$
for $i=1,\ldots , n$, and each $a_j$ with $j>n$ is $1$, $1-{1\over
5}$ or in $(0,{1\over 5}]$. Furthermore, there are at most five
items of size $1-{1\over 5}$. Let $m=\ceiling{n^{2\over b}}$.
Therefore, we have $m^b\ge n^2$. This makes us the sufficient
flexibility to select those items $a_i$ with $i>n$.  Let
$s=\sum_{i=1}^n a_i$. Select a number $n_1$ such that $m^b-5\le
(n_1-n)+s< m^b-4$. In other words, we have $m^b+n-s-5\le n_1
<m^b+n-s-4$. Thus, for all large $n$, we also have $n_1<m^b+n-s-4\le
m^b+n\le 2m^b<{m\over 2}$ since $b<1$. Let $a_i=1$ for all
$i=n+1,\ldots , n_1$. Therefore, $\sum_{i=1}^{n_1} a_i\in [m^b-5,
m^b-4)$. Then we select $a_i$ with $i=1,\ldots , m$ so that
$\sum_{i=1}^m a_i=m^b$. We select the next five items
$a_{n'+1},\ldots , a_{n'+5}$ of size $1-{1\over 5}$. Thus,
$\sum_{i=1}^{n_1} a_i\in [m^b-1, m^b)$. Let
$r=m^b-\sum_{i=1}^{n_1}a_i$. We have $r\in(0,1]$. The rest items
$a_{n'+6},a_{n'+7},\ldots , a_m$ are partitioned into five groups
$G_1,G_2, G_3, G_4$, and $G_5$ that size difference between any two
of them is at most one. Each item in $G_i$ is assigned ${r\over
5|G_i|}\in (0,{1\over 5}]$. Thus,  1) $\sum_{a_i\in G_j} a_i={r\over
5}$; 2)$\sum_{i=n'+6}^m a_i=r$; and 3) $\sum_{i=1}^m a_i=m^b$.

There is an optimal bin packing solution such that the five items of
size $1-{1\over 5}$ are in five bins with all items in the range
$(0,{1\over 5}]$. There is a solution for the 3-partition problem if
and only if the bin packing problem can be solved with ${n\over
3}+(n_1-n)+5$ bins. Any packing with ${n\over 3}+(n_1-n)+5$ bins for
$a_1,a_2,\ldots , a_m$ has to be the case that each item $a_j$ with
$j>n'+5$ is in a bin containing one item of size $1-{1\over 5}$
since it is impossible for $a_i$ ($i\le n$) to share a bin with
$a_j$ ($n'+1\le j\le n'+5$).
\end{proof}

Combining Theorem~\ref{NP-hard-sublinear-theorem} and
Theorem~\ref{dense-hierarchy-theorem}, we see a sublinear time
hierarchy of approximation scheme for a class of NP-hard problems,
which are derived from bin packing  problem. We show that the
$S(\delta)$-bin packing problem is NP-hard  if $\delta$ is at least
${1\over 4}$.

\begin{theorem}\label{delta-bin-packing-NP-theorem}
For each $\delta$ at most ${1\over 4}$, the $S(\delta)$-bin packing
problem is NP-hard.
\end{theorem}

\begin{proof}
We reduce the 3-partition problem to $S(\delta)$-bin packing
problem. Assume that $S=\{a_1,\ldots , a_{3m}\}$ is an input of
$3$-partition. We design that a $S(\delta)$-bin packing problem as
below: the bin size is $1$ and the items are ${a_1\over B},\ldots ,
{a_{3m}\over B}$. The size of each item is at least ${1\over 4}$
since each $a_i>{B\over 4}$. It is easy to see that there is a
solution for the $3$-partition problem if and only if those items
for the bin packing problem can be packed into $m$ bins.
\end{proof}

\section{Constant Time Approximation Scheme}

In this section, we show that there is a constant time approximation
for the $S(\delta)$-bin packing problem with $(c, \eta,k)$-related
bins for any positive constant $\delta$.

\begin{lemma}\label{core-lemma} Assume that $c$, $\eta$, and $k$ are constants.
Assume there is a $t(m,n,\mu)$ time and $z(m,n,\mu)$ queries
algorithm $A$ such that given a list $S$ of items of size at least
$\delta$, it returns $m$ items $y_1', y_2', \ldots , y_m'$ from the
list with $\Rank(y_i',S)\cap [ih-\mu h, ih+\mu h]\not=\emptyset$ for
$i=1,2,\ldots , m$. Then there is an $z(m,n,\mu)$ queries and
$t(m,n,\mu)+({1\over \epsilon\delta})^{O({1\over \delta})}$ time
approximation scheme $B$ for the $S(\delta)$-bin packing problem
with $(c,\eta,k)$-related bins. Furthermore, if $A$ fails with
probability at most $\alpha$, then $B$ also fails with probability
$\alpha$.

\end{lemma}

\begin{proof} Assume that $c,\eta$, and $k$ are positive constants.
Let $\epsilon$ be an arbitrary positive constant.  The constants
$\mu, \epsilon_1$, and $m$ are given according to equations
(\ref{mu-def}) to (\ref{m-def}).
We let the number of elements $n$ be large enough such that
${2q\over n\delta\eta}<{\epsilon\over 3}$, where $q$ is defined at
Lemma~\ref{lp-lemma}.

Assume that $a_1'\le a_2'\le \ldots \le a_{n_{\ge \delta}}'$ is the
increasing order of all input elements at least $\delta$. Let
$L_0=a_1'\le a_2'\le \ldots \le a_n'$.  We partition them into
$y_0A_1y_1A_2y_2\ldots  A_m y_mR$ such that each $A_i$ has exactly
$h$ elements and $R$ has less than $h$ elements.

Using algorithm $A$, we make approximation $y_i'$ to $y_i$ such that
the rank of $y_i'$ has at most $\mu h$ distance with that of $y_i$.
Assume that $\Rank(y_i',S)\cap [ih-\mu h, ih+\mu h]\not=\emptyset$
for $i=1,2,\ldots , m$ from algorithm $A$.


By Lemma~\ref{lp-lemma}, we have approximation scheme for
$\{y_1'^h,\ldots , y_m'^h\}$ with computational time $({1\over
\epsilon\delta})^{O({1\over \delta})}$, which follows from
Lemma~\ref{lp-lemma} and the selection of $m$ and $\mu $. The
approximation scheme for $S(\delta)$-bin packing problem follows
from Lemma~\ref{convert-approximation-lemma}. The total time is
$t(m,n,\mu)+({1\over \epsilon\delta})^{O({1\over \delta})}$ for
running $A$ and time involved in the algorithm of
Lemma~\ref{lp-lemma}.
\end{proof}

Lemma~\ref{core-lemma} is applied in both deterministic and
randomized algorithms in this paper. We note that algorithm $A$ in
Lemma~\ref{core-lemma} is deterministic if $\alpha=0$.

For the bin packing problem with item of size at least a positive
constant, our Theorem~\ref{constant-random-approx-theorem}
generalizes a result in~\cite{BatuBerenbrinkSohler09}.

\begin{theorem}\label{constant-random-approx-theorem} Assume that $c$, $\eta$, and $k$ are constants.
There is an $O({1\over \delta^2\epsilon^4})$ queries and $({1\over
\epsilon\delta})^{O({1\over \delta})}$ time randomized approximation
scheme algorithm for the $S(\delta)$-bin packing problem with
$(c,\eta,k)$-related bins.
\end{theorem}


\begin{proof}
Let $S$ be the list of input items of size at least $\delta$. It
follows from Lemma~\ref{select3-lemma}, and Lemma~\ref{core-lemma}.
By Lemma~\ref{select3-lemma}, we have a $t(m,n,\mu)=O({m^2(\log
m)^2)\over \mu^2})$ time algorithm such that using
$z(m,n,\mu)=O({m^2\log m\over \mu^2})$ random elements from $A$, it
generates elements $y_1'\le \ldots \le y_m'$ from the input list
such that $\prob[\Rank(y_i',S)\cap [ih-\mu h,ih+\mu h]=\emptyset$
for at least one $i\in \{1,\ldots , m\}]\le \alpha$.  We assume that
the $m$ items $y_1', y_2', \ldots , y_m'$ satisfy $\Rank(y_i',S)\cap
[ih-\mu h, ih+\mu h]\not=\emptyset$ for $i=1,2,\ldots , m$. The
approximation scheme follows from Lemma~\ref{core-lemma}.
\end{proof}

\begin{corollary}[\cite{BatuBerenbrinkSohler09}]\label{old-algorithm-corollary}
There is an $O({1\over \delta^2\epsilon^4})$ queries and $({1\over
\epsilon\delta})^{O({1\over \delta})}$ time approximation scheme
algorithm for the $S(\delta)$-bin packing problem.
\end{corollary}


We have Theorem~\ref{NP-constant-approximation2} that shows an
example of NP-hard problem that has a constant time approximation
scheme.

\begin{theorem}\label{NP-constant-approximation2}
There is an NP-hard problem that has a constant time approximation
scheme.
\end{theorem}

\begin{proof}
It follows from Theorem~\ref{delta-bin-packing-NP-theorem} and
Corollary~\ref{old-algorithm-corollary}.
\end{proof}

\section{Streaming Approximation Scheme}\label{streaming-sec}
In this section, we show a constant time and constant space
streaming algorithm for the bin packing problem. For the streaming
model of the bin packing  problem, we output a plan to pack the
items that have come from the input list, and the number of bins to
approximate the optimal number of bins. Our algorithm only holds a
constant number of items. Therefore, it has a constant updating time
and constant space complexity.

\begin{lemma}\label{select2-lemma}
There is an $O(u)$ updating time algorithm to select $u$ random
elements from a stream of input elements.
\end{lemma}

\begin{proof}
We set up $u$ positions to put the $u$ elements. There is a counter
$n$ to count the total number of elements arrived. For each new
arrived element $a_n$, the $j$-th position uses probability ${1\over
n}$ to replace the old element at the $j$-th position with the new
element. For each element $a_i$, with probability ${1\over j}{j\over
j+1}\ldots  {n-1\over n}={1\over n}$, it is kept at each of the $u$
positions after processing $n$ elements.  Therefore, we keep
$u$-random elements from the input list.
\end{proof}

A brief description of our streaming algorithm for the bin packing
problem is given in section~\ref{idea-overview-sec}.  Using the
method of Lemma~\ref{select2-lemma}, we maintain a list $X$ of
$O(1)$ random items of large sizes from the input list. The list is
updated after receiving every new element. The sizes of each small
item is added into a variable $s_1$. Using the method in
section~\ref{select-crucial-sec}, we find the approximate crucial
items from the list $X$ of random large items, which are the
approximate $ih$-th elements among the large items of size at least
$\delta$ for $i=1,\ldots , m$, where $h$ and $m$ are defined in
equations (\ref{h-def-eqn}), and ((\ref{m-def})), respectively. The
algorithm described in section~\ref{deterministic-section} is used
to pack large items. The small items are filled into bins which have
space left after packing large items, and some additional fresh
bins. With the sum $s_1$ of sizes of small items, we can calculate
the approximate number of fresh bins to be needed to pack them.

\vskip 10pt

 {\bf Algorithm Streaming-Bin-Packing}

Input: a positive constant $\epsilon$, and a streaming of items of
size at least $\delta$.

Output: an $(1+\epsilon)$-approximation.

 Steps:

\begin{enumerate}[1.]
\item
\qquad Let $\beta:={\gamma\over 30}$ and $\epsilon:=6\beta$.

\item
\qquad Let $\delta:={\epsilon\over 4}$ and $\theta:=0$.

\item
\qquad Let $\mu,\epsilon_1$ and $m$ are selected by
equations~(\ref{mu-def}), (\ref{epsilon1-def}), and (\ref{m-def}),
respectively.

\item
\qquad Let $c:=\eta:=k:=1$ (classical bin packing).

\item
\qquad Let $\alpha:=1/8$.

\item
 \qquad Let $u:={c_1m^2\log m\over \mu^2}$, where $c_1$ is
defined in Lemma~\ref{select3-lemma}.

\item
 \qquad Let $v:={2m\over \beta\delta}+m$.

\item
 \qquad Let $X[1...u]$ be an array of $u$ elements.

\item
 \qquad Let $X[i]:=0$ for $i=1,\ldots , u$.

 \item
 \qquad Let $Y[1...v]$ be an array of $v$ elements.

\item
 \qquad Let $Y[i]:=0$ for $i=1,\ldots , v$.

\item
\qquad Let $n:=0$.

\item
\qquad Let $n_{\ge\delta}:=0$.

\item
\qquad Let $s_1:=0$.

\item
\qquad For each new element $a_i$

\item
\qquad\qquad Let $n:=n+1$.

\item
\qquad\qquad If $a_i<\delta$

\item
\qquad\qquad then

\item
 \qquad\qquad\qquad Let $s_1:=s_1+a_i$.

\item
\qquad\qquad else

\item
\qquad\qquad\qquad Let $n_{\ge \delta}:=n_{\ge \delta}+1$.

\qquad\qquad\qquad If $n_{\ge \delta}<v$ then let $Y[n_{\ge
\delta}]:=a_i$.

\item
 \qquad\qquad\qquad For $i=1$ to $u$, let each $X[i]$ take the
new elements with probability ${1\over n_{\ge \delta}}$.

\item\label{if-condition-streaming-line}
\qquad\qquad If $n_{\ge\delta}> v$

\item
 \qquad\qquad then

\item
\qquad\qquad\qquad Output
Packing-With-Many-Large-Items$(\alpha,\beta, n, s_1, n_{\ge
\delta},S)$ (see Lemma~\ref{convert-approximation2b-lemma}).

\item
\qquad\qquad else

\item
\qquad\qquad\qquad Let $(b_1,\ldots ,
b_t)=$Linear-Time-Packing$(n_{\ge\varphi},Y)$ (see
Lemma~\ref{linear-time-packing}) (each bin $b_i$ represents a
packing of items in $Y$).

\item
 \qquad\qquad\qquad  For each $b_i$ with left space $u_i>\delta$,

\item
\qquad\qquad\qquad\qquad move $u_i-\delta$ (fractional)item size
into $b_i$ from $s_1$, and let $s_1=s_1-(u_i-\delta)$.

\item
\qquad\qquad\qquad Allocate $s_1$ into fresh bins such that each bin
except the last one wastes $\delta$ space.

\end{enumerate}

 {\bf End of Algorithm}

\vskip 10pt

\begin{theorem}\label{streaming-theorem}
Streaming-Bin-Packing is a  single pass  randomized streaming
approximation scheme for the  bin packing  problem such that it has
$O(1)$ updating time and $O(1)$ space, and computes an approximate
packing solution $Apx(n)$ with $\Sopt(n)\le App(n)\le
(1+\epsilon)\Sopt(n)+1$ in
$({1\over\epsilon})^{O({1\over\epsilon})}$ time, where $\Sopt(n)$ is
the optimal solution for the first $n$ items in the input stream,
and $App(n)$ is an approximate solution for the first $n$ items in
the input stream.
\end{theorem}


\begin{proof}
Let $\epsilon$ be an arbitrary positive constant. Let
$\delta={\epsilon\over 1+\epsilon}$.
 By
Lemma~\ref{select2-lemma}, we assume that $u$ random elements have
been selected from the input elements with size at least $\delta>0$.
We just add all elements with size less than $\delta$ into a sum
$s_1$.

If the condition of line~\ref{if-condition-streaming-line} in the
algorithm Streaming-Bin-Packing is true, then the
inequality~(\ref{large-n-ineqn}) in
Lemma~\ref{convert-approximation2b-lemma} can be satisfied since
$h'=\floor{n_{\ge \delta}\over m}$. Furthermore, as $\theta=0$, the
conditions of Lemma~\ref{convert-approximation2b-lemma} are
satisfied. The approximation ratio follows from
Lemma~\ref{convert-approximation2b-lemma}.

Assume the condition of line~\ref{if-condition-streaming-line} is
not true in the rest of the proof.
 Let $U$ be the set of bins for an $(1+\epsilon)$-approximate
solution to items of size at least $\delta$ by
Lemma~\ref{linear-time-packing}. It takes only $O(m)$ bins to pack
those large items since $n_{\ge \varphi}$ is less than $v$ which is
$O(m)$. Therefore, we only need $t=O(m)$ bins for packing the items
in $Y$. The final part of the algorithm fills all small items
accumulated in $s_1$ into those bins in $U$ so that each bin has
less than $\delta$ left. Put all of the items less than $\delta$
into some extra bins, and at most one of them has more than $\delta$
space left. Filling the small items of size less than $\delta$ is to
let each bin except the last one waste no more than $\delta$ space.
This is a fractional way to pack small items. Since the item size is
at most $\delta$, and each bin with (fractional) small items has at
least $\delta$ space left. The fractional bin packing for adding
small items can bring an non-fractional (regular) bin packing.  A
similar argument is also shown in Lemma~\ref{global-convert-lemma}.

Assume that an optimal solution of a bin packing  problem has two
types of bins. Each of the first type bin contains at least one item
of size $\delta$, and each of the second type bin only contains
items of size less than $\delta$. Let $V_1$ be the set of first type
bins, and $V_2$ be the set of all second type bins. We have that
$|U|\le (1+\epsilon)|V_1|$.

Case 1. If $U$ can contain all items, we have that $|U|\le
(1+\epsilon)|V_1|\le (1+\epsilon)|V_1\cup V_2|$.

Case 2. There is a bin beyond those in $U$ is used. Let $U'$ be all
bins without more than $\delta$ space left. We have that $|U'|\le
{|V_1\cup V_2|\over (1-\delta)}\le (1+\epsilon)|V_1\cup V_2|$.
Therefore, the approximate solution is at most $(1+\epsilon)|V_1\cup
V_2|+1$.

\end{proof}

\section{Sliding Window Streaming 
for $S(\delta)$-Bin
Packing}\label{sliding-windows-sect}

A sliding window stream model for bin packing problem is to pack the
most recent $n$ items. Select an integer constant $\lambda$ that is
determined by the approximation ratio and $\delta$, the least size
of input items. The idea is to start a new session to collect some
random items from the input stream after every ${n\over \lambda}$
items.

Assume that $a_{m+1},\ldots , a_{m+n}$ are the last $n$ input items
in the input stream. We maintain a list of sets $S_1,\ldots ,
S_{\lambda}$ such that if $S_i$ is a set of random items in
$\{a_{m+j_i},\ldots , a_{m+n}\}$ ($[m+j_i, m+n]$ is called the range
of $S_i$), then the next $S_{(i+1)(mod\ \lambda)}$ is a set of
random items in $\{a_{m+j_i+{n\over \lambda}},\ldots , a_{m+n}\}$.
On the other hand, when the range of a set $S_i$ reaches
$[m+1,m+n]$), $S_i$ is reset to be empty and starts to collect the
random elements from the scratch. We also set a pointer to the set
$S_i$ that has the largest range.

After receiving every ${n\over \lambda}$ items in the input stream,
the set $S_i$ with the largest range will be passed to the next
$S_{i+1(\mod \ \lambda)}$ if $S_i$'s range reaches size $n$. The is
called rotation, which makes the pointer to the set with the largest
range  according to the loop $S_1\rightarrow S_2\rightarrow, \ldots
, S_{\lambda-1}\rightarrow S_{\lambda}\rightarrow S_1$. In the
following algorithm we assume that $n=0\ (\mod \ \lambda)$.
Otherwise, we replace $n$ by $n'=\ceiling{n\over \lambda}\lambda$.

It is easy to see that $n\le n'\le n+\lambda$. The bin packing
problem for the last $n$ items has a small ratio difference with
that for the last $n'$ items if the constant $\lambda$ is selected
large enough.


\vskip 10pt

 {\bf Algorithm Sliding-Window-Bin-Packing($c, \eta, k,\gamma, \delta, n$)}

Input: bin types constants $c, \eta$, and $k$, a positive constant
$\gamma$, a streaming of items of size at least $\delta$, and a
sliding window size $n$.

Output: an $(1+\gamma)$-approximation

 Steps:

\begin{enumerate}[1.]

\item
\qquad Let $\epsilon:={\gamma\over 30}$.   

\item
\qquad Let $\theta:=0$.

\item
\qquad Let $\mu,\epsilon_1$ and $m$ are selected by
equations~(\ref{mu-def}), (\ref{epsilon1-def}), and (\ref{m-def}),
respectively.

\item\label{set-k-line}
 \qquad Let $\lambda:=\ceiling{100\over \gamma\delta}$.

\item
\qquad Let $t:={n\over \lambda}$.

\item
\qquad Create $t$ empty sets $S_1,\ldots , S_k$ to hold random
elements and make them non-active.

\item
\qquad Let $u:={c_1m^2\log m\over \mu^2}$ be the number of random
elements in each $S_i$ according to Lemma~\ref{select3-lemma}.

\item
\qquad Let $h_{j}$ be the range size of $S_j$.

\item
\qquad Start $S_1$ to be active to collect random elements.

\item
\qquad Let $S_1$ hold $u$ copies of the first element $a_1$ in the
stream.

\item
\qquad For each new element $a_i$ from the input stream
($i=2,3,\ldots $)

\item
\qquad\qquad For each active $S_j$ and each of the $u$ items $a_r\in
S_j$,

\item
\qquad\qquad\qquad replace $a_r$ by $a_i$ with probability ${1\over
h_{j}}$ and let $h_i:=h_i+1$.

\item
\qquad\qquad Let $S_j$ be the set with the largest range $h_j$.

\item
\qquad\qquad Let $(y_1,\ldots , y_m)$:=Select-Crucial-Items($m, u,
S_j$) (see Lemma~\ref{select3-lemma}).

\item
\qquad\qquad Let $(x_1,\ldots , x_q):=$Pack-Large-Items($c, \eta, k,
B$) with $B=\{y_1'^{h},\ldots , y_m'^{h}\}$ (See
Lemma~\ref{lp-lemma}).

\item
\qquad\qquad Let $y$ be the cost for the packing with solution
$(x_1,\ldots , x_q)$.

\item
\qquad\qquad Output $app:=$Packing-Conversion($n,y$) (see
Lemma~\ref{convert-approximation-lemma}).

\item\label{start-if-sliding-line}
\qquad\qquad If $i=0 (mod\ t)$

\item
\qquad\qquad then

\item
\qquad\qquad \qquad if $i<n$

\item
\qquad\qquad \qquad  then make $S_{(j+1)\ (\mod\ t)}$ be active, and
let $h_{(j+1)(mod\ t)}=0$.

\item\label{end-if-sliding-line}
\qquad\qquad\qquad if $i\ge n$

\item
\qquad\qquad\qquad then let $S_j$ hold $u$ copies of $a_i$ and let
$h_{j}=1$ (reset $S_j$).

\end{enumerate}

{\bf End of Algorithm}

\vskip 10pt

We have Theorem~\ref{NP-constant-sliding} that shows an example of
NP-hard problem that has a constant time and constant space sliding
window streaming approximation scheme.

\begin{theorem}\label{sliding-window-theorem}
  Assume that $c$, $\eta$, and $k$ are constants. Let $\delta$ be an arbitrary constant. Then
Sliding-Window-Bin-Packing(.) is a  single pass sliding window
streaming randomized approximation algorithm for the $S(\delta)$-bin
packing problem with $(c,\eta,k)$-related bins that has $O(1)$
updating time and $O(1)$ space, and computes an approximate packing
solution $App(.)$ with $\Sopt_{c,\eta,k}(n)\le App(n)\le
(1+\gamma)\Sopt_{c,\eta,k}(n)$ in
$({1\over\gamma})^{O({1\over\gamma})}$ time, where
$\Sopt_{c,\eta,k}(n)$ is the optimal solution for the last $n$ items
in the input stream, and $App(n)$ is an approximate solution for the
most recent $n$ items in the input stream.
\end{theorem}

\begin{proof}
Multiple sessions of groups are generated to maintain the progress of incoming elements.
 The purpose of the choice of $\lambda$ at line~\ref{set-k-line} in Sliding-Window-Bin-Packing(.) is to let it satisfy that
${n\over \lambda}\le {\gamma n\delta/100}$ since $n$ items needs at
least $m\delta$ bins and ${n\over \lambda}$ items needs at most
${n\over  \lambda}$ bins. This control is implemented in
lines~\ref{start-if-sliding-line} to \ref{end-if-sliding-line} in
the algorithm Sliding-Window-Bin-Packing(.).

Assume that the $n$ integers in $[1,n]$ represent the last $n$ items
from the input stream. Each of the $\lambda$ groups takes care of
the list items in the range $[i\cdot {n\over \lambda}+j, n]$ for
$i=0,\ldots , \lambda-1$, where $j$ is an integer that moves in the
loop $0\rightarrow 1\rightarrow 2\rightarrow 3\rightarrow \ldots
\rightarrow {n\over \lambda}-1 \rightarrow 0$. We keep $\lambda$
groups of $u={c_1m^2\log m\over \mu^2}$ random elements each
according to Lemma~\ref{select2-lemma}, where $\mu$ is defined as
the proof in Lemma~\ref{core-lemma} and Lemma~\ref{select3-lemma}.
After every ${n\over \lambda}$ items, we start picking a new session
of elements and drop the oldest session.

When a set $S_j$ holds $u$ random elements from the last $h$
elements for $h\in [n-t, n+t]$, where $t={n\over \lambda}$. The
approximation derived from $S_j$ has a small difference with the
optimal solution for the last $n$ elements. Let
$\Sopt_{c,\eta,k}(n)$ be the optimal solution for packing the last
$n$ items with $(c,\eta,k)$-related bins. We have that
$\Sopt_{c,\eta,k}(h)-t\le \Sopt_{c,\eta,k}(n)\le
\Sopt_{c,\eta,k}(h)+ t$. By Lemmas~\ref{lp-lemma},
\ref{select3-lemma}, and \ref{convert-approximation-lemma}, the
algorithm outputs an $(1+\gamma/2)$-approximation for
$\Sopt_{c,\eta,k}(h)$. By the setting of $t$, we have that
$(1-\gamma/2)\Sopt_{c,\eta,k}(h)\le \Sopt_{c,\eta,k}(n)\le
(1+\gamma/2)\Sopt_{c,\eta,k}(h)$. Therefore, an $(1+\gamma/2)$
approximation to $\Sopt_{c,\eta,k}(h)$ is a $(1+\gamma)$ to
$\Sopt_{c,\eta,k}(n)$.

\end{proof}

\begin{theorem}\label{NP-constant-sliding}
There is an NP-hard problem that has a constant time and space
sliding windows approximation scheme.
\end{theorem}

\begin{proof}
It follows from Theorem~\ref{delta-bin-packing-NP-theorem} and
Theorem~\ref{sliding-window-theorem}.
\end{proof}

\subsection{Constant Time Approximation Scheme for Random
Sizes}\label{random-input-sec}
 In this section, we identify more
cases of the bin packing problem with constant time approximation.
One interesting case is that all items are random numbers in
$(0,1]$.

\begin{definition} Let $\delta_1,\delta_2$ and $\epsilon_1$ are positive parameters.
For a list $a_1,\ldots , a_n$ of input of bin packing problem, it
has the  {\it $(\delta_1,\delta_2,\epsilon_1)$-property} if the list
 $a_1,\ldots , a_n$ satisfies $$\ceiling{{\delta_2\over c-\delta_2}{|\{i: a_i\le \delta_2\ and \ a_i\in
\{a_1,\ldots , a_n\}\}|}}\le \epsilon_1\eta\delta_1{|\{i: a_i\ge
\delta_1\ and \ a_i\in \{a_1,\ldots , a_n\}\}|}.$$
\end{definition}

\begin{theorem}\label{most-constant-theorem}
Let $\delta_1,\delta_2$ and $\epsilon$ are positive constants with
$\delta_2\ge\delta_1$. Then there is a constant $({1\over
\epsilon\delta_1})^{O({1\over\delta_1})}$ time algorithm such that
if the bin packing problem with $(c,\eta,k)$-related bins and
$(\delta_1,\delta_2,\epsilon/3)$-property, it gives an
$(1+\epsilon)$-approximation.
\end{theorem}


\begin{proof}
Let $t_1$ be the cost of an optimal solution to pack those items of
size at least $\delta_1$ and $t_2$ be the cost of an optimal
solution to pack those items of size at most $\delta_2$. Let $t$ be
the cost of an optimal solution to pack all items in the list.
Clearly, we have $t\ge t_1$.

The number of bins is at least $b_1=\delta_1 |\{i: a_i\ge \delta_1\
and \ a_i\in \{a_1,\ldots , a_n\}\}|$ for packing those items of
size at least $\delta_1$. The cost for packing those items of size
at least $\delta_1$ is at least $\eta b_1$ since the least cost is
$\eta$ among all bins. Thus, $\eta b_1\le t_1$. The number of bins
for packing those items of size at most $\delta_2$ is at most
$b_2=\ceiling{{\delta_2\over c-\delta_2}{|\{i: a_i\le \delta_2\ and
\ a_i\in \{a_1,\ldots , a_n\}\}|}}$ since at most one bin wastes
space more than $\delta_2$. The cost for packing those items of size
$\delta_2$ is at most $b_2$ since $1$ is the upper bound of the
largest cost bin.

With $({1\over \epsilon\delta_1})^{O({1\over\delta_1})}$ time, we
derive an $(1+{\epsilon\over 3})$-approximation $b_1'$ for the items
of size at least $\delta_1$ by
Theorem~\ref{constant-random-approx-theorem}. We have $b_1\le b_1'$
since $b_1'$ is an approximation to the optimal solution and $b_1$
is a lower bound of the optimal solution for packing items of size
at least $\delta_1$. The cost for the bins for packing those items
of size at most $\delta_2$ is at most $b_2\le {\epsilon\over 3}\eta
b_1\le {\epsilon\over 3}\eta b_1'$ because of the
$(\delta_1,\delta_2,\epsilon/3)$-property. We output the
approximation with cost $b_1'+{\epsilon\over 3}\eta b_1'$. We have
\begin{eqnarray*}
b_1'+ {\epsilon\over 3}\eta b_1'&\le& (1+{\epsilon\over
3})t_1+{\epsilon\over 3}(1+{\epsilon\over 3})t_1\ \ \ \ \ \mbox{(note\ $\eta\le 1$)}\\
&\le& (1+\epsilon)t_1\\
&\le& (1+\epsilon)t.
\end{eqnarray*}
Therefore, we derive an $(1+\epsilon)$-approximation for packing the
input list with $(c,\eta,k)$-related bins.

\end{proof}


\begin{theorem}\label{constant-time-corollary}  Assume that $c$, $\eta$, and $k$ are constants.
Assume that $a$ and $b$ with $a<b\le c$ are two constants in
$[0,1]$. Let $\epsilon$ be a constant in $(0,1]$. Then there is a
randomized constant $({1\over \epsilon})^{O({1\over (a+\epsilon)})}$
time approximation scheme for the bin packing problem with
$(c,\eta,k)$-related bins
 that each element is a random element
from $[a,b]$.
\end{theorem}


\begin{proof}
Let $\epsilon_2$ be a constant in $(0,{1\over 4})$ and will be
determined later. Let $\epsilon_1={\epsilon\over 3}$. Let
$\delta_1=\delta_2=a+{\epsilon_2(b-a)}$. We prove that a list with
random elements from $[a,b]$  satisfies
$(\delta_1,\delta_2,\epsilon_1)$-property for all large $n$ with
high probability. Assume that $a_1,\ldots , a_n$ is a list of random
elements in $[a,b]$.

We note that with probability $0$, a random element $a_i$ from
$[a,b]$ is equal to $a$. For each random element $a_i\in [a,b]$,
with probability $p_1=1-\epsilon_2$, we have $a_i\ge \delta_1$.  By
Theorem~\ref{chernoff3-theorem}, with probability at most
$P_1=g_1({1\over 4})^{p_1n}$, $n_1=|\{i: a_i\ge \delta_1\}|$ is less
than $(p_1-{1\over 4})n$ elements. We note $(p_1-{1\over 4})n\ge
{n\over 4}$ since $p_1\ge {1\over 2}$.

For each random element $a_i\in [a,b]$, with probability
$p_2=\epsilon_2$, we have $a_i< \delta_2$. By
Theorem~\ref{ourchernoff2-theorem}, with probability at most
$P_2=g_2(1)^{p_2n}$, we have $n_2=|\{i: a_i< \delta_2\}|$ is more
than $(1+1)p_2n=2\epsilon_2 n$.

Assume that  $n_1\ge {n\over 4}$ and $n_2\le 2\epsilon_2 n$.

Since $\epsilon_2$ is a constant in $(0,{1\over 4})$, we have
$\delta_2\le a+{1\over 4}(b-a)$. Thus,
 we have
\begin{eqnarray*}
{\delta_2\over c-\delta_2}&\le&{\delta_2\over {b-\delta_2}}\\
&\le&{b\over {b-\delta_2}}\\
&\le&{b\over {b-(a+{1\over 4}(b-a))}}\\
&\le& {4b\over 3(b-a)}.
\end{eqnarray*}
Assume that $n$ is large enough such that $({1\over b-a})\epsilon_2
n\ge 1$.
 We have that
\begin{eqnarray*}
 \ceiling{{\delta_2\over c-\delta_2}
n_2}&\le& {\delta_2\over c-\delta_2} n_2+1\\
&\le& {4b\over 3(b-a)} n_2+1\\
&\le& {4b\over 3(b-a)}\cdot 2\epsilon_2n+1\\
&\le& {16b\over 3(b-a)}\epsilon_2n\\
&\le&\epsilon_1\eta\delta_1 {n\over 4}\\
 &\le &\epsilon_1\eta\delta_1
n_1,
\end{eqnarray*}
 where  $\epsilon_2$ is selected to be ${3\epsilon_1\eta\delta_1(b-a)\over 64b}$, which is less than ${1\over 4}$.
Therefore, with probability at most $P_1+P_2$, the
$(\delta_1,\delta_2,\epsilon_1)$-property is not satisfied.
Theorem~\ref{constant-time-corollary} follows from
Theorem~\ref{most-constant-theorem}.
\end{proof}

\begin{theorem}\label{constant-time2-corollary} Assume that $a<b$ are two constants in $[0,1]$. Then there is a
randomized constant $({1\over \epsilon})^{O({1\over a+\epsilon})}$
time approximate scheme for the bin packing problem
 that each element is a random element
from $[a,b]$.
\end{theorem}

\begin{proof}
It follows from Theorem~\ref{constant-time-corollary}.
\end{proof}

\section{Conclusions}

This paper shows a dense hierarchy of approximation schemes for the
bin packing problem which has a long history of research. Pursing
sublinear time algorithm brings a better understanding about the
technology of randomization, and also gives some new insights about
the problems that may already have linear time solution. Our
sublinear time algorithms are based on an adaptive random sampling
method for the bin packing problem developed in this paper.
The hierarchy approach, which is often used in the complexity
theory, may give a new way for algorithm analysis as it gives more
information than the worst case analysis from the classification.

\section{Acknowledgements}

We would like to thank Xin Han for his helpful suggestions which
improves the presentation of this paper. This research is supported
in part by National Science Foundation Early Career Award 0845376.
An earlier version of this paper is posted at
http://arxiv.org/abs/1007.1260.









\end{document}

\section{ Proofs for Chernoff Bounds}

For completeness, we give the proofs for the versions of  Chernoff
bounds used  in this paper. We follow the proof of
Theorem~\ref{chernoff-theorem} to make the following version of
Chernoff bound so that it can be used in our
algorithm analysis. 


\subsection{Proof for Theorem~\ref{ourchernoff2-theorem}}

\begin{proof}
 Let $t$ be an arbitrary positive real number. By the
definition of expectation, we have
$E(e^{tX_i})=\Pr(X_i=1)e^t+\Pr(X_i=0)$. Since the function
$f(x)=xe^t+(1-x)$ is increasing for all $t>0$ and $\Pr(X_i=1)\le p$,
we have $E(e^{tX_i})\le pe^t+(1-p)$. We have the following
inequalities:
\begin{eqnarray}
\Pr(X>(1+\delta)pn)&<&{E(e^{tX})\over e^{t(1+\delta)pn}}\label{markov-inequ}\\
&\le&{\prod_{i=1}^nE(e^{tX_i})\over e^{t(1+\delta)pn}}\label{independent1}\\
&=&{\prod_{i=1}^n(pe^t+1-p)\over e^{t(1+\delta)pn}}\label{independent2}\\
&=&{\prod_{i=1}^n(1+p(e^t-1))\over e^{t(1+\delta)pn}}\\
&\le&{\prod_{i=1}^ne^{p(e^t-1)}\over e^{t(1+\delta)pn}}\\
&=&{e^{(e^t-1)pn}\over e^{t(1+\delta)pn}}\\
&=&({e^{(e^t-1)}\over e^{t(1+\delta)}})^{pn}\label{final1}.
\end{eqnarray}
The inequality (\ref{markov-inequ}) is based on Markov inequality.
The transition from (\ref{independent1}) to (\ref{independent2}) is
due to the independence of those variables $X_1,\ldots , X_n$.

Since $({e^{(e^t-1)}\over e^{t(1+\delta)}})$ is minimal at
$t=\ln(1+\delta)$, we have
$\Pr(X>(1+\delta)pn)<\left[{e^{\delta}\over
(1+\delta)^{(1+\delta)}}\right]^{pn}$.
\end{proof}

\subsection{Proof for Theorem~\ref{chernoff3-theorem}}


\begin{proof}
$\prob[X<(1-\delta)pn]=\prob[-X>-(1-\delta)pn]=\prob[e^{-yX}>e^{-y(1-\delta)pn}]$
for each real number $y$. Applying Markov inequality, we have
\begin{eqnarray}
\prob[X<(1-\delta)pn]&<&{\prod_{i=1}^nE(e^{-yX_i}]\over
e^{-y(1-\delta)np}}\label{chernoff-1}\\
&<&{e^{(e^{-y}-1)np}\over e^{-y(1-\delta)np}}\label{chernoff-2}\\
&<&\left( {e^{-\delta}\over
(1-\delta)^{1-\delta}}\right)^{pn}\label{chernoff-3}\\
&<&e^{-{1\over 2}^{pn}\delta^2}.\label{chernoff-4}
\end{eqnarray}
The transition from~(\ref{chernoff-2}) to (\ref{chernoff-2}) is to
let $t=\ln {1\over 1-\delta}$. The transition
from~(\ref{chernoff-3}) to (\ref{chernoff-4}) follows from the fact
$(1-\delta)^{1-\delta}>e^{-\delta+\delta^2/2}$.
\end{proof}




\subsection{Proof for Corollary~\ref{chernoff-lemma-a}}




\begin{proof}
For $X=\sum_{i=1}^n$, $\mu=E(X)=\sum_{i=1}^nE(X_i)=pn$. Let
$\delta={\epsilon\over p}$. (1) follows from
Theorem~\ref{chernoff-theorem}.  By Taylor theorem,
$\ln(1+\epsilon)\ge \epsilon-{\epsilon^2\over 2}$. We have that
$(1+{1\over \epsilon})\ln(1+\epsilon)\ge(1+{1\over
\epsilon})(\epsilon-{\epsilon^2\over 2})=1+{\epsilon\over
2}-{\epsilon^2\over 2}>1+{\epsilon\over 3}$. Thus, $ {e\over
(1+\epsilon)^{(1+{1\over \epsilon})}}<e^{-{\epsilon\over 3}}$. Since
$pn+\epsilon n=(1+\delta)\mu$ and the function $(1+y)^{1\over y}$ is
increasing for $y>0$, $Pr(X>pn+\epsilon
n)=Pr(X>(1+\delta)\mu)<\left[ {e^{\epsilon\over p}\over
(1+{\epsilon\over p})^{(1+{\epsilon\over p})}}\right]^{pn}=\left[
{e\over (1+{\epsilon\over p})^{(1+{p\over
\epsilon})}}\right]^{\epsilon n}\le \left[ {e\over
(1+\epsilon)^{(1+{1\over \epsilon})}}\right]^{\epsilon n}\le
e^{-{\epsilon^2n\over 3}}.$ Thus (ii) is proved.
\end{proof}

\end{document}